\documentclass[pdflatex,sn-mathphys-num]{sn-jnl}

\usepackage[T1]{fontenc}
\usepackage[utf8]{inputenc}
\usepackage{graphicx}
\usepackage{amsmath,amssymb,amsfonts}
\usepackage{amsthm}
\usepackage{booktabs}
\usepackage{float}

\theoremstyle{thmstyleone}
\newtheorem{theorem}{Theorem}
\newtheorem{proposition}{Proposition}
\newtheorem{corollary}{Corollary}

\theoremstyle{thmstyletwo}

\newtheorem{remark}{Remark}

\theoremstyle{thmstylethree}

\begin{document}

\title[]{A Conjugate Bayesian Framework for Fast 3D Positronium Lifetime Estimation with a Partial System Matrix}

\author[1]{\fnm{Berkin} \sur{Uluutku}}
\author[1]{\fnm{Giulianno} \sur{Gasparato}}
\author[2]{\fnm{Manish} \sur{Das}}
\author[3]{\fnm{Jaros{\l}aw} \sur{Choi{\'n}ski}}
\author[2]{\fnm{Anand} \sur{Pandey}}
\author[2]{\fnm{Sushil} \sur{Sharma}}
\author[2]{\fnm{Pawe{\l}} \sur{Moskal}}
\author[2]{\fnm{Ewa} \sur{St{\k e}pie{\'n}}}
\author[4]{\fnm{Chien-Min} \sur{Kao}}
\author*[1]{\fnm{Hsin-Hsiung} \sur{Huang}}\email{hsin.huang@ucf.edu}

\affil[1]{\orgdiv{School of Data, Mathematical, and Statistical Sciences}, \orgname{University of Central Florida}, \orgaddress{\city{Orlando}, \state{FL}, \country{USA}}}
\affil[2]{\orgdiv{Faculty of Physics, Astronomy and Applied Computer Science and Center for Theranostics}, \orgname{Jagiellonian University}, \orgaddress{\city{Krakow}, \country{Poland}}}
\affil[3]{\orgdiv{Heavy Ion Laboratory}, \orgname{University of Warsaw}, \orgaddress{\city{Warsaw}, \country{Poland}}}
\affil[4]{\orgdiv{Department of Radiology}, \orgname{The University of Chicago}, \orgaddress{\city{Chicago}, \state{IL}, \country{USA}}}

\abstract{
\textbf{Background:} Positronium lifetime imaging extends conventional positron emission tomography by using the time interval between positron emission and annihilation as an additional contrast mechanism. Voxel-wise lifetime estimation in fully three-dimensional settings is computationally difficult because the number of feasible detector-time channels grows rapidly, whereas only a small subset is observed in practice. We developed a scalable statistical framework for three-dimensional positronium lifetime estimation based on a time-of-flight-aware partial system matrix restricted to observed detector-time channels, combined with posterior event-to-voxel weighting and a conjugate Gamma--Exponential update for closed-form voxel-wise effective-rate estimation.

\textbf{Results:} Restricting the forward model to observed detector-time channels reduced memory and computational requirements while preserving the Poisson data model for retained detected triple coincidences. In simulated data with 4056 voxels, the analytic Bayesian estimator required 2.76~s versus 74.46~s for 10 L-BFGS-B iterations on the same CPU while accurately recovering the effective-rate map. In a triple-coincidence dataset acquired with a J-PET prototype scanner and a NEMA image-quality phantom, a 234\,375-voxel effective-rate map was estimated in approximately 3~s from about \(3.64\times10^5\) retained events.

\textbf{Conclusions:} Restricting the system matrix to observed detector-time channels makes fully three-dimensional positronium lifetime estimation computationally practical for sparse triple-coincidence data. The proposed posterior-weighted conjugate update provides a fast and stable single-component surrogate estimator of voxel-wise effective lifetime for large-scale three-dimensional positronium lifetime imaging.
}


\maketitle

\section{Background}\label{sec:intro}

Positronium lifetime imaging (PLI) extends conventional positron emission tomography (PET) by exploiting the time interval between positron emission and annihilation as an additional source of image contrast \citep{moskal_positronium_2021}. Because positronium properties depend on the surrounding microenvironment, the measured lifetime carries information that complements conventional activity imaging \citep{moskal_positronium_2019-1}. Over the past several years, experimental demonstrations of positronium imaging have progressed rapidly, including dedicated detector systems \citep{moskal_positronium_2021,huang_high-resolution_2025,samanta_feasibility_2023,takyu_positronium_2024}, high-resolution phantom studies \citep{huang_high-resolution_2025,das_first_2025, das_first_2026,kubat_first_2026}, and recent in vivo measurements on clinical and prototype scanners \citep{moskal_positronium_2024,mercolli_vivo_2026}. These developments motivate reconstruction methods capable of estimating voxel-wise lifetime parameters in fully three-dimensional settings. Reconstruction methods for positronium imaging now include direct statistical formulations for time-of-flight (TOF) PET \cite{qi_positronium_2022}, surrogate likelihood methods such as SPLIT \cite{huang_split_2024}, and fast high-resolution reconstruction strategies such as SIMPLE \cite{huang_fast_2025}. At the same time, recent experimental studies on commercial scanners have emphasized that realistic lifetime analysis is influenced by multi-component decay behavior, detector timing response, and random triple coincidences \cite{steinberger_positronium_2024,mercolli_first_2025}.

The main challenge addressed in this work is efficient three-dimensional estimation when only a small fraction of feasible detector–time channels is actually observed in the data. We consider a two-stage estimation strategy. First, we estimate a detected activity image representing the expected number of retained detected triple coincidences contributed by each voxel. Second, conditional on this activity estimate, we estimate a voxel-wise effective annihilation-rate map using a single-component exponential model. The resulting rate should be interpreted as a compact representation of the observed timing signal rather than a complete multi-component model of positronium decay.

The paper makes three contributions. First, we formulate a TOF-aware partial system matrix defined only on observed detector–time channels, substantially reducing memory and computational cost while preserving the Poisson data model. Second, we introduce an event-to-voxel weighting formulation that interprets the unknown source voxel of each event as a latent variable, enabling principled fractional attribution of events across voxels. Third, we introduce a conjugate Bayesian estimator that yields closed-form voxel-wise annihilation-rate estimates, enabling fast three-dimensional estimation while providing a compact mean representation of the observed positronium lifetime signal.

Finally, the proposed framework is evaluated using both simulated data and an experimental triple-coincidence dataset acquired with a J-PET prototype scanner and a NEMA Image-Quality phantom \citep{das_first_2025}. The results demonstrate that the partial system matrix formulation enables computationally tractable three-dimensional reconstruction while the conjugate Bayesian estimator provides stable lifetime estimates and voxel-wise uncertainty information in realistic experimental conditions.

\section{Methods}\label{sec:meth}

This section presents the 3D estimation framework and simulation setup. To make volumetric estimation computationally feasible, we replace the full system matrix with a partial system matrix defined only on observed detector--time channels. We then estimate voxel-wise annihilation rates using both maximum likelihood and a conjugate Bayesian update, the latter providing fast inference and voxel-wise uncertainty estimates. Figure~\ref{fig:flow} summarizes the proposed workflow. 

\begin{figure}[h]
\centering
\includegraphics[width=0.9\textwidth]{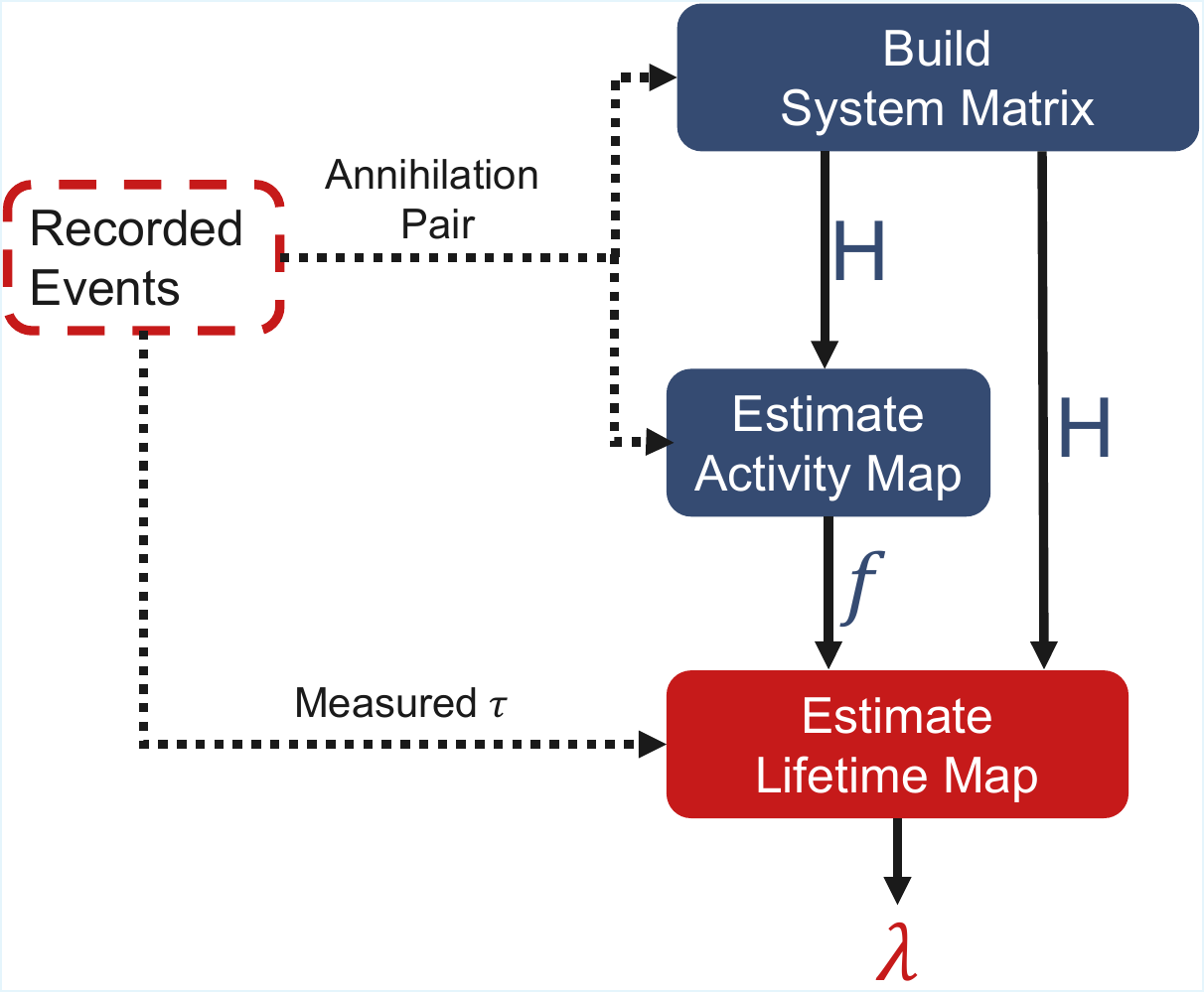}
\caption{Proposed estimation workflow}\label{fig:flow}
\end{figure}

\subsection{Event Representation and Preprocessing}\label{ssec:preprocess}

Each recorded decay is treated as a triple coincidence consisting of two annihilation-photon interactions and one prompt-gamma interaction. In list mode, event \(k\) is stored as
\begin{equation}
 e_k = \{d_{1,k}, d_{2,k}, t_{1,k}, t_{2,k}, d_{p,k}, t_{p,k}\},
\label{eq:eventtuple}
\end{equation}
where \(d_{1,k}\) and \(d_{2,k}\) are the detectors for the annihilation photons in arrival order, \(d_{p,k}\) is the detector for the prompt gamma ray, and \(t_{1,k}, t_{2,k}, t_{p,k}\) are the corresponding detection times.

The annihilation-pair TOF observable is
\begin{equation}
\Delta t_k = t_{2,k} - t_{1,k}.
\end{equation}
Because the annihilation photons are indexed in arrival order, \(\Delta t_k\ge 0\). With TOF-bin width \(\Delta\tau_{\mathrm{bin}}\) and TOF window \(\Delta t_{\max}\), the TOF bin index is
\begin{equation}
\xi_k = \min\!\left(\left\lfloor \frac{\Delta t_k}{\Delta \tau_{\mathrm{bin}}} \right\rfloor,\, N_{\xi}-1\right),
\qquad
N_{\xi}=\left\lceil \frac{\Delta t_{\max}}{\Delta\tau_{\mathrm{bin}}} \right\rceil.
\label{eq:tofbin}
\end{equation}
We denote the ordered annihilation-detector pair by \(\eta_k=(d_{1,k},d_{2,k})\) and the detector-time channel by \(c_k=(\eta_k,\xi_k)\).

The event-wise prompt-to-annihilation time difference is
\begin{equation}
\tau_k = t_{a,k} - t_{p,k},
\qquad
t_{a,k}=\frac{t_{1,k}+t_{2,k}}{2}.
\label{eq:eventlifetime}
\end{equation}
Because the voxel of origin is unknown, \(\tau_k\) contains both the desired lifetime information and an unmodeled flight-path mismatch. In the present single-component exponential surrogate model, we retain only events with \(\tau_k\ge 0\). This filtering removed 3.6\% of simulated detected triples. In the experimental dataset it removed 20.6\% of all recorded triples and 8.3\% of events within the estimated \(^{44}\)Sc regions. The larger fraction among all recorded triples is consistent with enrichment of background or mismatched events among negative-\(\tau_k\) observations. More flexible exponentially modified Gaussian type models should absorb rather than discard such events in future work \cite{steinberger_positronium_2024,mercolli_vivo_2026}. In what follows, \(N_{\mathrm{ev}}\) denotes the number of retained events after this \(\tau_k\ge 0\) filter.

\subsection{Partial System Matrix and Detected Activity}\label{ssec:partial_system}

Let \(\mathcal{C}\) denote the feasible set of detector-time channels implied by the scanner geometry, coincidence window, and TOF binning. Because the annihilation detectors are indexed in arrival order, the total number of such channels grows on the order of \(N_{\mathrm{sensor}}(N_{\mathrm{sensor}}-1)N_{\xi}\), which becomes computationally expensive even in sparse form. In finite acquisitions, however, only a small subset of these channels is observed.

We therefore restrict the forward model to the observed channels,
\begin{equation}
\mathcal{C}_+ = \{c_k : k=1,\dots,N_{\mathrm{ev}}\}
          = \{c\in \mathcal{C}: y_c > 0\},
\end{equation}
where \(y_c\) is the count in channel \(c\). The partial system matrix is defined as the conditional distribution
\begin{equation}
H_{c,j} = P(C=c\mid J=j,\; C\in\mathcal{C}_+,J\in\mathcal{V}), \qquad c\in\mathcal{C}_+,
\label{eq:partialH}
\end{equation} where \(\mathcal{V} = \{1,\ldots,N_{\mathrm{vox}}\}\) is the voxel domain. Under this definition,
\begin{equation}
\sum_{c\in\mathcal{C}_+} H_{c,j} = 1, \qquad j\in\mathcal{V},
\end{equation}
so \(H\) is normalized over observed detector-time channels for each voxel.

Under this parameterization, \(\tilde f_j\) denotes the expected number of retained detected triple coincidences contributed by voxel \(j\) during the acquisition. The observed channel counts satisfy the Poisson model
\begin{equation}
y_c \sim \mathrm{Poisson}(\hat y_c),
\qquad
\hat y_c = \sum_{j=1}^{N_{\mathrm{vox}}} H_{c,j} \, \tilde f_j,
\qquad c\in\mathcal{C}_+.
\label{eq:poissonchannel}
\end{equation}
Thus, activity reconstruction in this manuscript targets the distribution of \emph{retained detected} triples, not the latent total number of all physical decays. In a realistic 3D scanner, many emissions escape the detector geometry and fail to produce triple coincidences. The system matrix is therefore conditioned on retained detected events. Since \(\tilde f_j\) already represents retained detected activity, the source-voxel probability for a randomly selected retained event is proportional to \(\tilde f_j\), whereas the factor \(H_{c,j}\) enters only after conditioning on the observed channel \(c\).

\subsection{Detected Activity Reconstruction}

We estimate \(\tilde f=\{\tilde f_j\}\) by maximizing the Poisson likelihood in Eq.~\eqref{eq:poissonchannel} with expectation-maximization (EM) \cite{shepp_maximum_1982}. For iterate \(n\),
\begin{equation}
\tilde f_j^{(n+1)} = \tilde f_j^{(n)}
\sum_{c\in\mathcal{C}_+} H_{c,j}\frac{y_c}{\hat y_c^{(n)}},
\qquad
\hat y_c^{(n)} = \sum_{j=1}^{N_{\mathrm{vox}}} H_{c,j}\tilde f_j^{(n)}.
\label{eq:emupdate}
\end{equation}
Because \(\sum_{c\in\mathcal{C}_+}H_{c,j}=1\) by construction, the usual EM sensitivity factor is equal to one and does not appear explicitly in Eq.~\eqref{eq:emupdate}. Because \(H\) is stored only on \(\mathcal{C}_+\), each forward- and backprojection visits observed channels only. This is the principal computational advantage of the partial-system-matrix representation in sparse triple-coincidence settings.

\subsection{Effective annihilation-rate estimation}

After the EM iterations, we fix the detected-activity estimate \(\tilde f=\{\tilde f_j\}\) and estimate the voxel-wise annihilation-rate field \(\lambda=\{\lambda_j\}\) from the retained list-mode lifetimes \(\{\tau_k\}\). We consider two complementary estimators: (i) bound-constrained maximum likelihood using L-BFGS-B, and (ii) a conjugate Gamma update with fractional event attribution that provides a closed-form posterior  approximation.

\subsubsection{Posterior event-to-voxel weights}

For a randomly selected retained event, the marginal voxel-of-origin probability is
\begin{equation}
p_j=\frac{\tilde f_j}{\sum_{\ell=1}^{N_{\mathrm{vox}}}\tilde f_\ell}.
\label{eq:main_pj}
\end{equation}
Because \(\tilde f_j\) already denotes retained detected activity, Eq.~\eqref{eq:main_pj} is a marginal retained-event probability and does not include the channel factor \(H_{c,j}\).

For each retained event \(k\) with observed channel \(c_k\), define the unnormalized source score
\begin{equation}
w_{k,j}=H_{c_k,j}\tilde f_j \propto H_{c_k,j}p_j,
\label{eq:unnormalized_weights_main}
\end{equation}
and the normalized channel posterior weight
\begin{equation}
\pi_{k,j}=
\frac{w_{k,j}}{\sum_{\ell=1}^{N_{\mathrm{vox}}} w_{k,\ell}}
=
\frac{H_{c_k,j}\tilde f_j}{\sum_{\ell=1}^{N_{\mathrm{vox}}} H_{c_k,\ell}\tilde f_\ell}.
\label{eq:posterior_weights_main}
\end{equation}
Here \(\pi_{k,j}=P(J=j\mid C_k=c_k,\tilde f)\) is the posterior probability, under the detected-activity model alone, that event \(k\) originated from voxel \(j\). Consequently, \(\sum_{j=1}^{N_{\mathrm{vox}}}\pi_{k,j}=1\) for each retained event.

\subsubsection{Maximum-Likelihood Estimator}

Conditional on \(\tilde f\) and the observed channel \(c_k\), the retained-event lifetime density is
\begin{equation}
p(\tau_k\mid C_k=c_k,\tilde f,\lambda)
=
\sum_{j=1}^{N_{\mathrm{vox}}}
\pi_{k,j}\lambda_j e^{-\lambda_j\tau_k},
\qquad \tau_k\ge 0.
\end{equation}
Up to an additive constant that does not depend on \(\lambda\), the conditional log-likelihood and its gradient are
\begin{equation}
\begin{aligned}
\label{eq:loglik_r}
\ell(\lambda\mid \tilde f)
&=
\sum_{k=1}^{N_{\mathrm{ev}}}
\log\!\left(
\sum_{j=1}^{N_{\mathrm{vox}}}
\pi_{k,j}\lambda_j e^{-\lambda_j\tau_k}
\right),
\\[6pt]
\frac{\partial \ell}{\partial \lambda_j}
&=
\sum_{k=1}^{N_{\mathrm{ev}}}
r_{k,j}
\left(\frac{1}{\lambda_j}-\tau_k\right).
\end{aligned}
\end{equation}
where
\begin{equation}
r_{k,j}(\lambda)=
\frac{\pi_{k,j}\,\lambda_j e^{-\lambda_j\tau_k}}
{\sum_{\ell=1}^{N_{\mathrm{vox}}} \pi_{k,\ell}\,\lambda_{\ell}e^{-\lambda_{\ell}\tau_k}}
\label{eq:responsibilities}
\end{equation}
is the event-specific responsibility. We maximize Eq.~\eqref{eq:loglik_r} subject to \(\lambda_j\ge 0\) using L-BFGS-B \cite{virtanen_scipy_2020}. Because \(\pi_{k,j}\propto w_{k,j}\) with a channel-specific normalizing constant independent of \(\lambda\), this likelihood is equivalent to the formulation based on the unnormalized scores \(w_{k,j}\).

\subsubsection{Posterior-weighted Gamma approximation}

Since the lifetime follows an exponential distribution, an analytic conditional posterior is available through conjugate Bayesian inference when the voxel labels are known. To obtain a fast closed-form approximation, we place independent Gamma priors
\begin{equation}
\lambda_j\sim \mathrm{Gamma}(\alpha_0,\beta_0),
\qquad \alpha_0>0,\ \beta_0>0,
\end{equation}
using the shape-rate parameterization. If the latent labels \(Z_k\), namely the voxel of origin for each event, were known, the conditional posterior would remain Gamma with sufficient statistics
\begin{equation}
n_j=\sum_{k=1}^{N_{\mathrm{ev}}}I(Z_k=j),
\qquad
S_j=\sum_{k=1}^{N_{\mathrm{ev}}}I(Z_k=j)\tau_k.
\label{eq:pigfly}
\end{equation}

The latent labels are unknown, and their exact posterior probabilities after observing \(\tau_k\) depend on \(\lambda\). For a fast one-pass approximation, we replace the unknown indicators \(I(Z_k=j)\) by the normalized channel-posterior weights \(\pi_{k,j}=P(J=j\mid C_k=c_k,\tilde f)\), which are available once \(\tilde f\) has been reconstructed. This gives the effective sufficient statistics
\begin{equation}
n_j^{\mathrm{eff}}=\sum_{k=1}^{N_{\mathrm{ev}}}\pi_{k,j},
\qquad
S_j^{\mathrm{eff}}=\sum_{k=1}^{N_{\mathrm{ev}}}\pi_{k,j}\tau_k.
\label{eq:effstats_r}
\end{equation}
By construction, \(\sum_{j=1}^{N_{\mathrm{vox}}} n_j^{\mathrm{eff}} = N_{\mathrm{ev}}\).

The resulting posterior-weighted Gamma approximation is
\begin{equation}
\lambda_j\mid \mathcal D
\approx
\mathrm{Gamma}\!\left(\alpha_0+n_j^{\mathrm{eff}},\,\beta_0+S_j^{\mathrm{eff}}\right),
\label{eq:gammaapprox_r}
\end{equation}
with posterior mean
\begin{equation}
\widehat{\lambda}_j^{\mathrm{EB}}
=
\frac{\alpha_0+n_j^{\mathrm{eff}}}{\beta_0+S_j^{\mathrm{eff}}}.
\label{eq:postmean_r}
\end{equation}
For display we report the reciprocal posterior-mean rate
\begin{equation}
\widehat{\tau}_j^{\mathrm{eff}} =\frac{1}{\widehat{\lambda}_j^{\mathrm{EB}}}.
\label{eq:taudef_r}
\end{equation}
This approximation uses channel-conditioned posterior labels rather than the exact lifetime-dependent responsibilities \(r_{k,j}\), and is intended as a fast surrogate to the full mixture posterior.

\paragraph{Statistical interpretation.}
The proposed estimator can be viewed as a finite mixture model in which voxels are latent components and retained events are observations. The channel-posterior weights \(\pi_{k,j}\) encode voxel-of-origin probabilities based on detector-time information and the reconstructed detected-activity image, whereas the responsibilities \(r_{k,j}\) additionally incorporate the lifetime model. Under the idealized retained-event model and standard regularity assumptions, identifiability and large-sample consistency can be established for the fixed-channel, positive-support formulation of this model. Formal statements and derivations are provided in Additional file~1.

\subsection{Simulation study}

The simulation phantom occupied a \(26\times 26\times 6\) grid with voxel size \(1\times 1\times 3\)~cm\(^3\). The background was an ellipsoid centered at the origin with effective rate \(\lambda=0.5\)~ns\(^{-1}\). Four ellipsoidal inclusions, centered at \((\pm 4~\mathrm{cm},\pm 4~\mathrm{cm},0)\), had twice the background activity and effective rates \(0.4\), \(0.6\), \(0.8\), and \(1.0\)~ns\(^{-1}\). A total of \(100\times 10^6\) simulated decays was distributed according to these activity levels. Ground truth for generated events can be seen in Figures~\ref{fig:phantom_activity}~and~\ref{fig:phantom_lambda}.

\begin{figure}[htbp]
\centering{\includegraphics[width=4in]{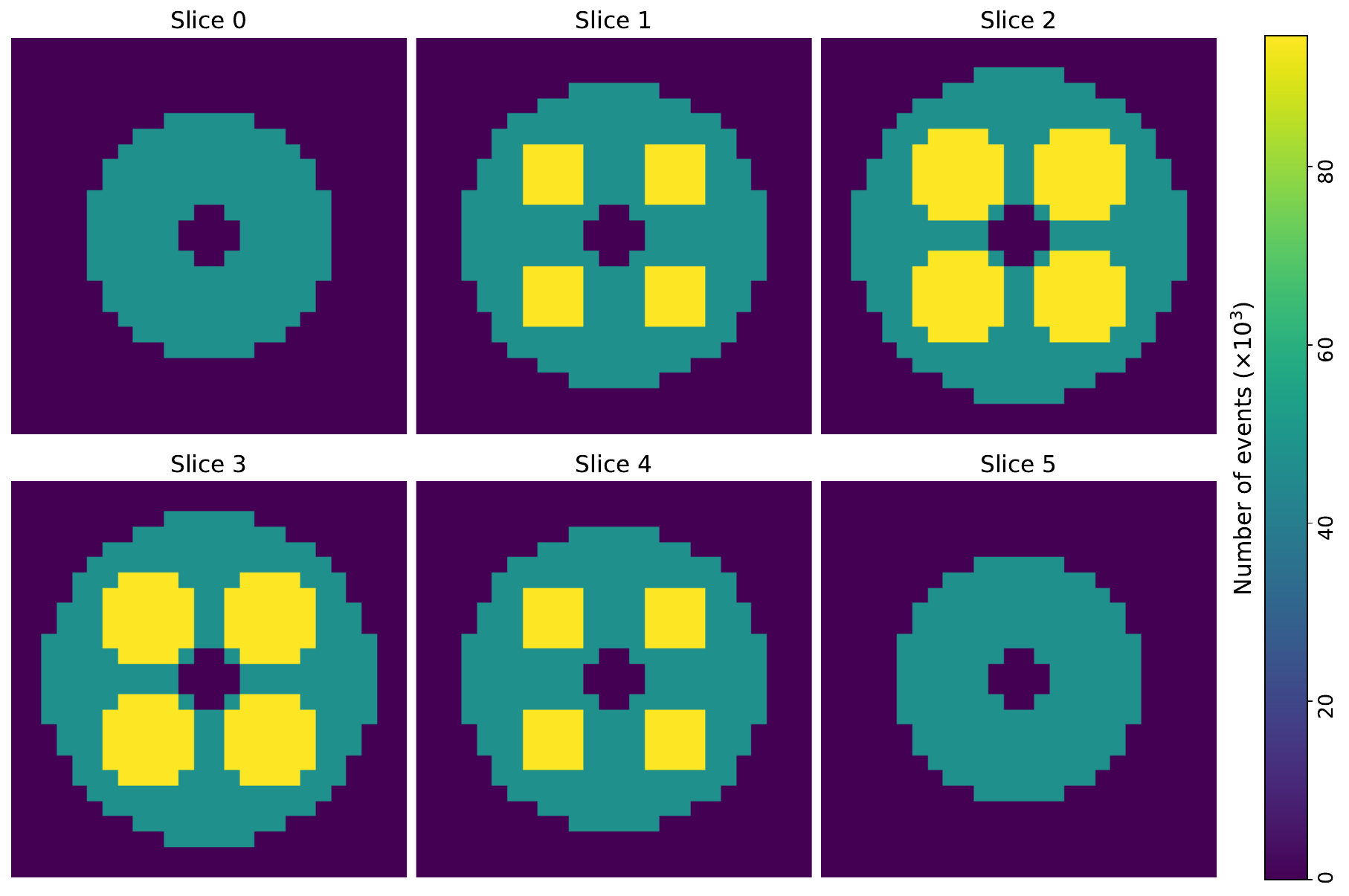}}
\caption{\label{fig:phantom_activity} Ground-truth activity map of the
simulated 3D phantom. The color scale represents the number of
annihilation events in thousands. The central region is the uniform
background; brighter regions mark higher-activity inclusions.}
\end{figure}

\begin{figure}[htbp]
\centering{\includegraphics[width=4 in]{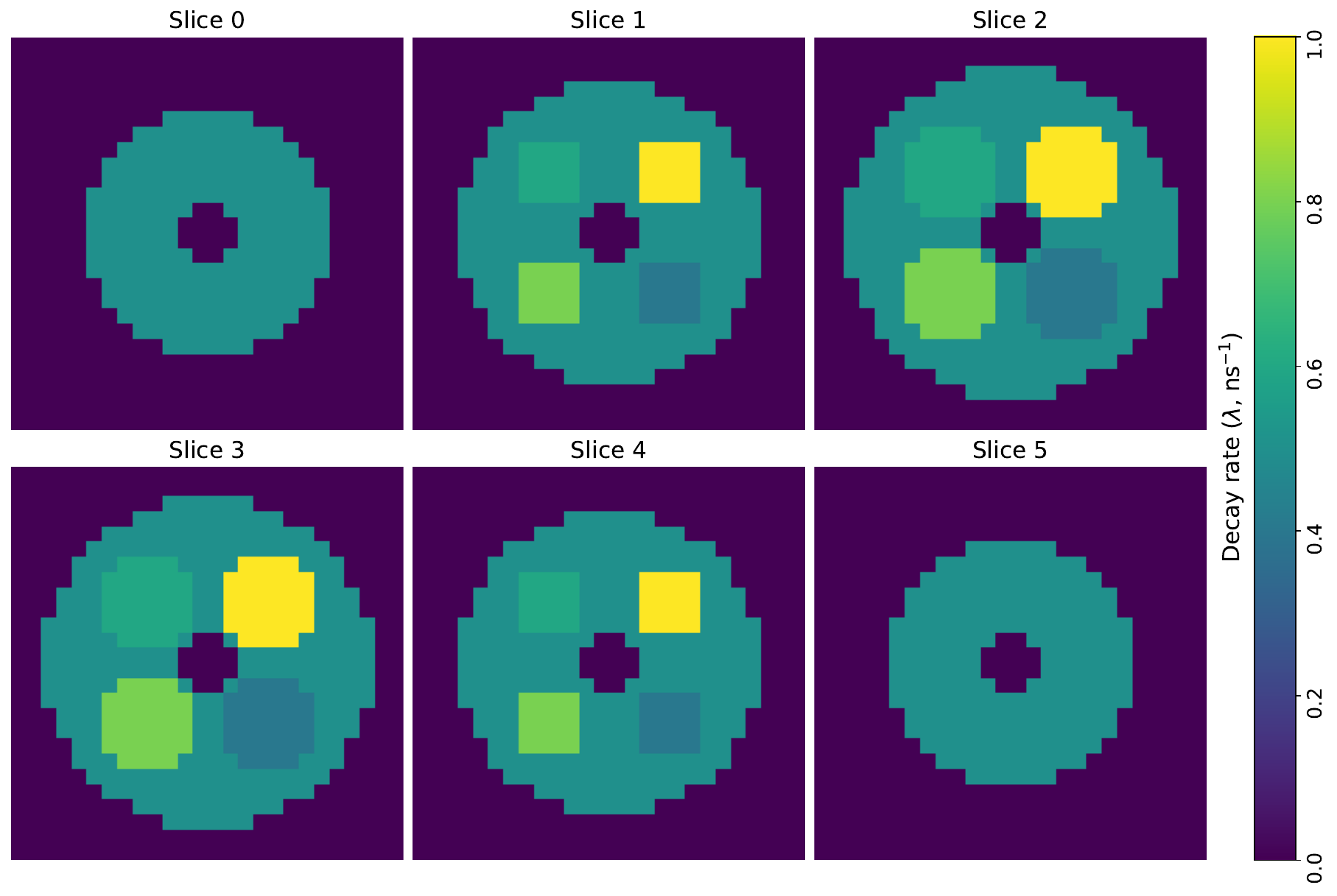}}
\caption{\label{fig:phantom_lambda} Ground-truth lifetime map. The color
scale shows positronium decay rates (ns$^{-1}$). Four inclusions with
increasing $\lambda$ are embedded in the uniform background to evaluate
sensitivity to lifetime variation.}
\end{figure}

The scanner model was cylindrical with diameter 60~cm, axial length 25~cm, 288 detectors per ring, and 6 axial rings, for a total of 1728 detector elements. The coincidence resolving time was set to 200~ps (FWHM). Event generation sampled emission points within voxels, isotropic directions for the prompt and annihilation photons, and an exponential lifetime with the voxel-wise effective rate. Photon propagation determined detector hits and arrival times.

Of the initial \(100\times 10^6\) decays, approximately \(4.6\times 10^6\) were accepted as detected triple coincidences and approximately \(4.5\times 10^6\) remained after removing events with \(\tau_k<0\). Ground truth for the detected activity can be seen in Figure~\ref{fig:captured_activity}.

\begin{figure}[H]
\centering{\includegraphics[width=4in]{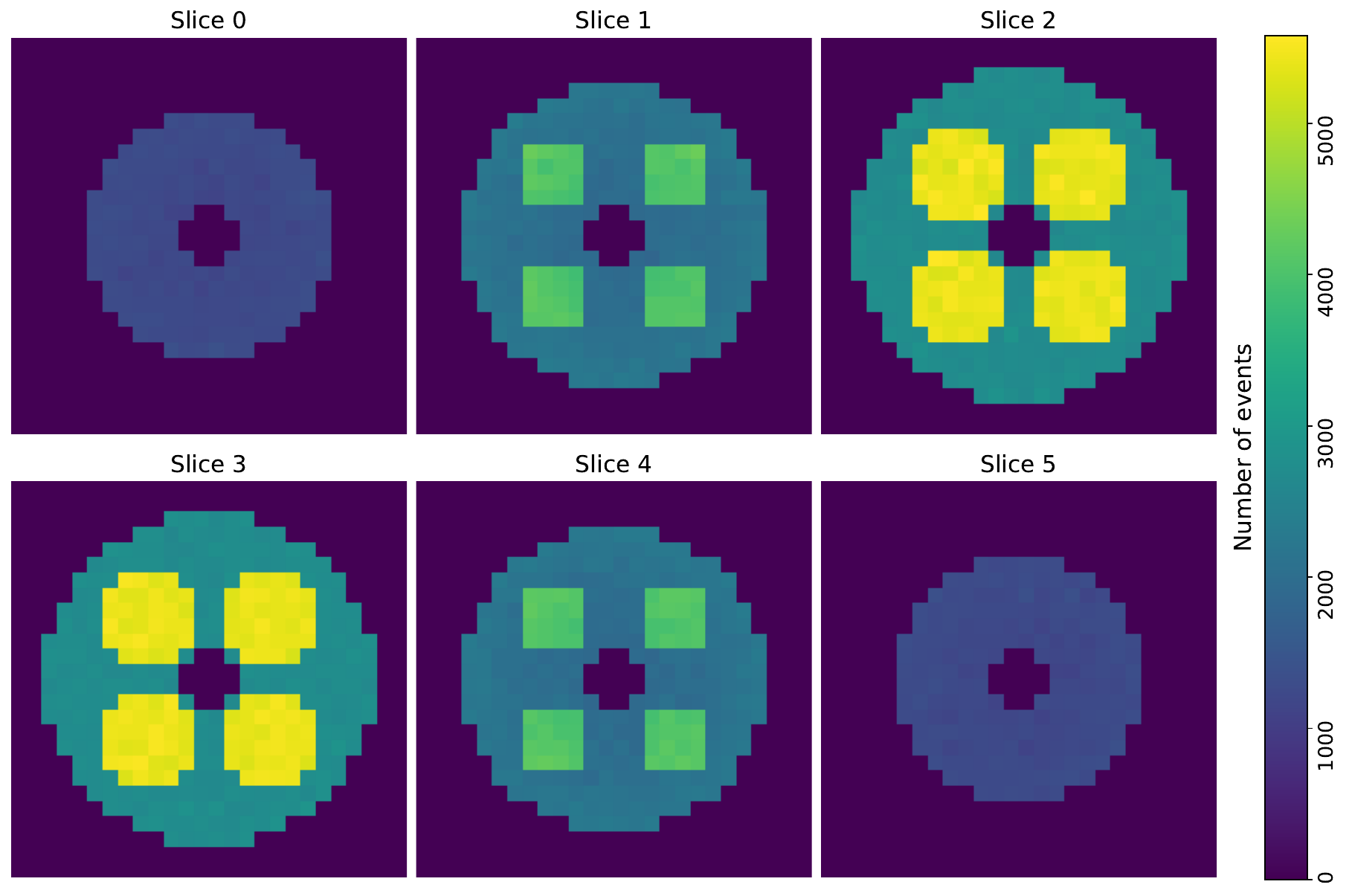}}
\caption{\label{fig:captured_activity} Detected activity after physics. Counts are retained triple coincidences per voxel and represent the data available to estimate.}
\end{figure}

Because the simulation was generated from the same single-component model used in estimation, lifetime recovery in this experiment measures estimator behavior under model match. To avoid conflating lifetime recovery with extreme low-count instability, rate-map comparisons were emphasized for voxels with at least 10 retained detected triples. The simulation and estimation settings are summarized in Table~\ref{tab:sim}.

\begin{table}[htbp]
\caption{Simulation setup and estimation summary}\label{tab:sim}
\begin{tabular}{@{}p{0.63\linewidth}p{0.25\linewidth}@{}}
\toprule
Parameter & Value \\
\midrule
Image grid & \(26\times26\times6\) voxels \\
Voxel size & \(1\times1\times3\) cm\(^3\) \\
Total simulated decays & \(100\times 10^6\) \\
Detected triple coincidences & \(\approx 5.0\times10^5\) \\
Retained events (\(\tau_k\ge 0\)) & \(\approx 4.8\times10^5\) \\
EM iterations & 5 \\
L-BFGS-B iterations & 10 \\
Bayesian hyperparameters & \(\alpha_0=\beta_0=10^{-4}\) \\
\botrule
\end{tabular}
\end{table}

\subsection{Experimental dataset}

The experimental evaluation used a 3D triple-coincidence dataset acquired with a modular J-PET prototype scanner \citep{tayefi_ardebili_assessing_2024} and a NEMA image-quality phantom. The three largest spheres were filled with \(^{44}\)Sc dissolved in water and the three smallest spheres were filled with \(^{18}\)F; the background compartment contained water without radioactivity. Acquisition details and phantom preparation are described in the dedicated experimental report by Das et al. \cite{das_first_2025}. Because $^{44}\mathrm{Sc}$ decay includes a prompt $\gamma$ emission, it produces true triple-coincidence events $(2\gamma_{511}+\gamma_p)$ required for PLI, whereas $^{18}\mathrm{F}$ produces only 511\,keV annihilation pairs and does not contribute true triples.

The estimation grid comprised \(125\times125\times15\) voxels with voxel size \(0.33\times0.33\times1\)~cm\(^3\). A total of 458\,639 triple-coincidence events was recorded, of which approximately \(3.64\times10^5\) remained after the \(\tau_k\ge 0\) filter. Activity was estimated with five EM iterations. The experimental setup and estimation summary are listed in Table~\ref{tab:exp}.

\begin{table}[htbp]
\caption{Experimental dataset summary}\label{tab:exp}
\begin{tabular}{@{}p{0.63\linewidth}p{0.25\linewidth}@{}}
\toprule
Parameter & Value \\
\midrule
Image grid & \(125\times125\times15\) voxels \\
Voxel size & \(0.33\times0.33\times1\) cm\(^3\) \\
Recorded triple coincidences & 458\,639 \\
Retained events (\(\tau_k\ge 0\)) & \(\approx 3.64\times10^5\) \\
Observed detector-time channels & \(\approx 6.2\times10^4\) \\
EM iterations & 5 \\
\botrule
\end{tabular}
\end{table}

The $^{44}$Sc spheres contain homogeneous activity, so no positronium lifetime contrast is expected between them. The reconstructed lifetime maps therefore demonstrate the behaviour of the reconstruction framework rather than contrast between different environments. As described in the methodology, the present analysis uses a simplified lifetime model and focuses on the reconstruction framework rather than a full physical characterization of positronium decay.

\section{Results}\label{sec:results}

This section is organized in two parts. We first present results from the simulation studies, followed by the corresponding analysis on the experimental dataset.

\subsection{Simulation Results}

Figure~\ref{fig:f_recreated} shows the reconstructed activity obtained
using five EM iterations applied to the partial system matrix. Because
the statistical model is defined for detected triple coincidences,
the reconstruction estimates the distribution of activity that produced the detected
data. Comparisons are therefore made to the detected activity map in
Figure~\ref{fig:captured_activity}. The reconstruction recovers the spatial structure of all regions and maintains the expected relative intensities across slices. Small deviations appear at the boundaries between regions of differing activity, where finite detector coverage and reduced statistics near the edges produce mild blurring.

\begin{figure}[htbp]
\centering{\includegraphics[width=4in]{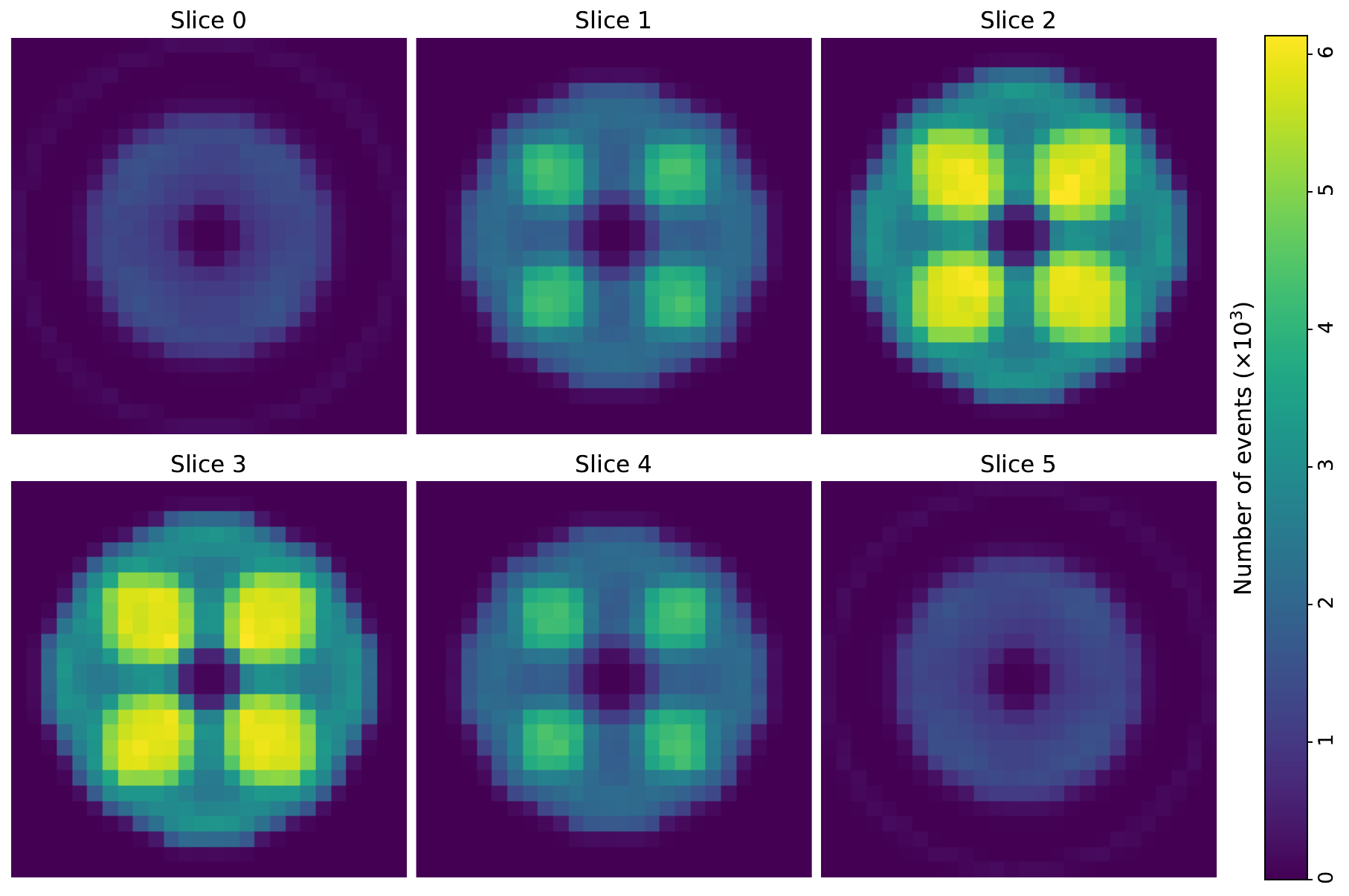}}
\caption{\label{fig:f_recreated} Reconstructed detected activity distribution across all six axial layers. Each slice shows consistent estimated triple-coincidence counts per voxel.}
\end{figure}

To quantify activity accuracy, Figure~\ref{fig:boxplot_ferror} reports
voxel-wise normalized absolute errors between the reconstructed and
detected activity maps, grouped by ground-truth activity level. Median errors are near zero across all groups, indicating limited bias. The high-activity group shows a slightly broader distribution and a small downward shift, consistent with boundary blur and higher variance from fewer counts concentrated in compact inclusions. As discussed, these estimates correspond to detected events rather than total emissions. A geometric approximation for total-decay recovery is provided in Additional file~1.

\begin{figure}[H]
\centering{\includegraphics[width=4in]{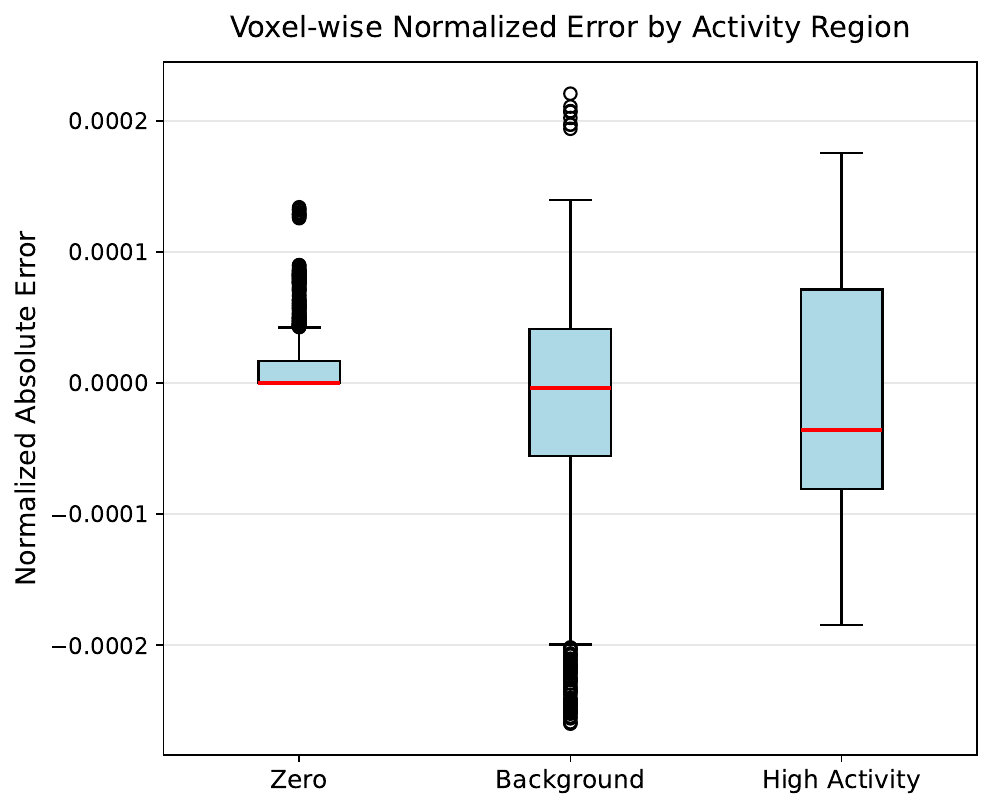}}
\caption{\label{fig:boxplot_ferror} Distribution of voxel-wise normalized absolute errors between reconstructed and detected activity, grouped by activity level (zero, background, high). Boxes summarize the spread and bias within each region.}
\end{figure}

For effective rate estimation, both L-BFGS-B and the posterior-weighted Gamma approximation produced highly accurate reconstructions in the simulation. The posterior-weighted conjugate Bayesian estimator shown in Figure~\ref{fig:lambda_conjugate} reproduces the same spatial pattern and dynamic range of effective annihilation rates. Minor smoothing at region interfaces is consistent with the averaging nature of fractional attribution and the grid resolution. Figure~\ref{fig:lambda_map_compare} compares voxel-wise relative errors for both lifetime estimators across the four inclusion regions. Median errors are near zero for all groups. The Bayesian method exhibits tighter distributions, particularly at \(\lambda=0.5\), indicating reduced voxel-level variance. At higher decay rates (\(\lambda \ge 0.8\)), both methods show slight underestimation, consistent with fewer detected events in shorter-lifetime regions. In the 4056-voxel simulation, L-BFGS-B (10 iterations) required 74.46~s, whereas the analytic Bayesian estimator required 2.76~s on the same CPU. Image results of the L-BFGS-B estimator are provided in Additional file~1.

\begin{figure}[htbp]
\centering{\includegraphics[width=4in]{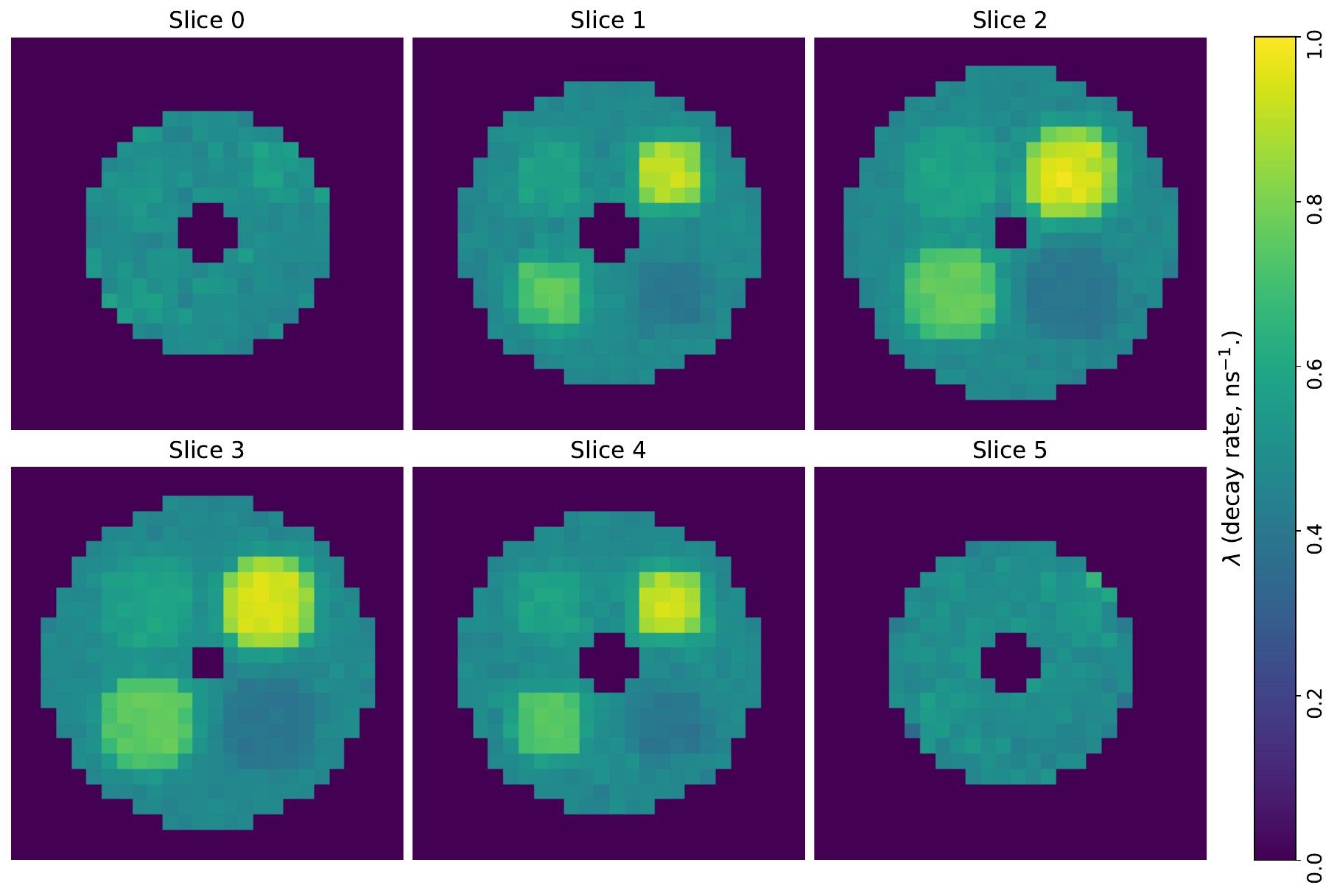}}
\caption{\label{fig:lambda_conjugate} Voxel-wise effective annihilation-rate map reconstructed with the posterior-weighted conjugate Bayesian update. Spatial contrast and magnitude agree with the ground truth.}
\end{figure}

\begin{figure}[htbp]
\centering{\includegraphics[width=4in]{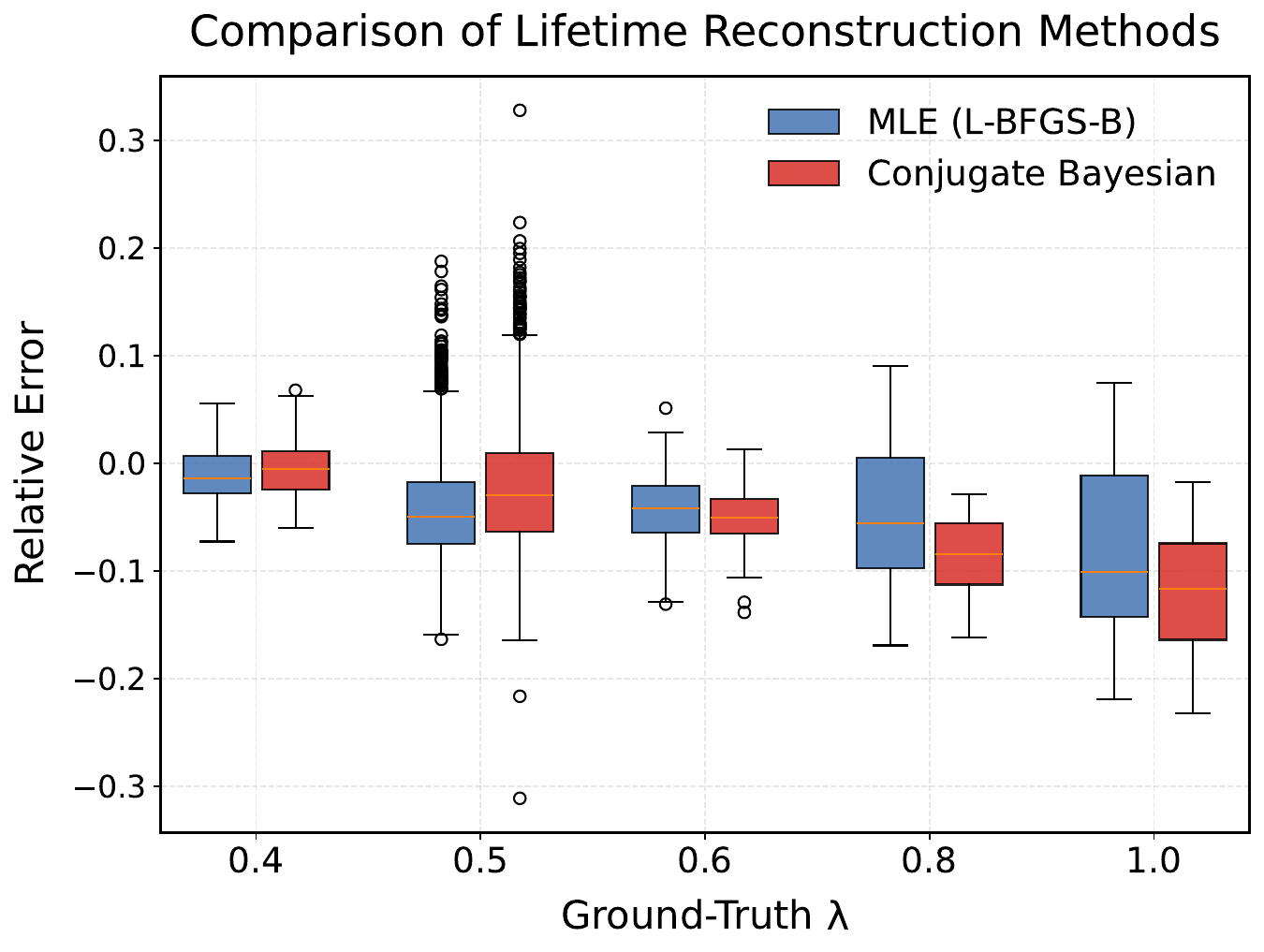}}
\caption{\label{fig:lambda_map_compare} Voxel-wise relative error of reconstructed decay rates across ground-truth \(\lambda\) regions, comparing L-BFGS-B and conjugate Bayesian estimators.}
\end{figure}

\subsection{Experimental Results}

The estimated detected-activity image is shown in Figure~\ref{fig:jactivity}. The three \(^{44}\)Sc-filled spheres are visible in the central slices, and low-level activity appears in nominally cold regions. This background is consistent with false or misassociated triples, which is expected in the absence of an explicit random-triple correction step.

\begin{figure}[htbp]
\centering
\includegraphics[width=4in]{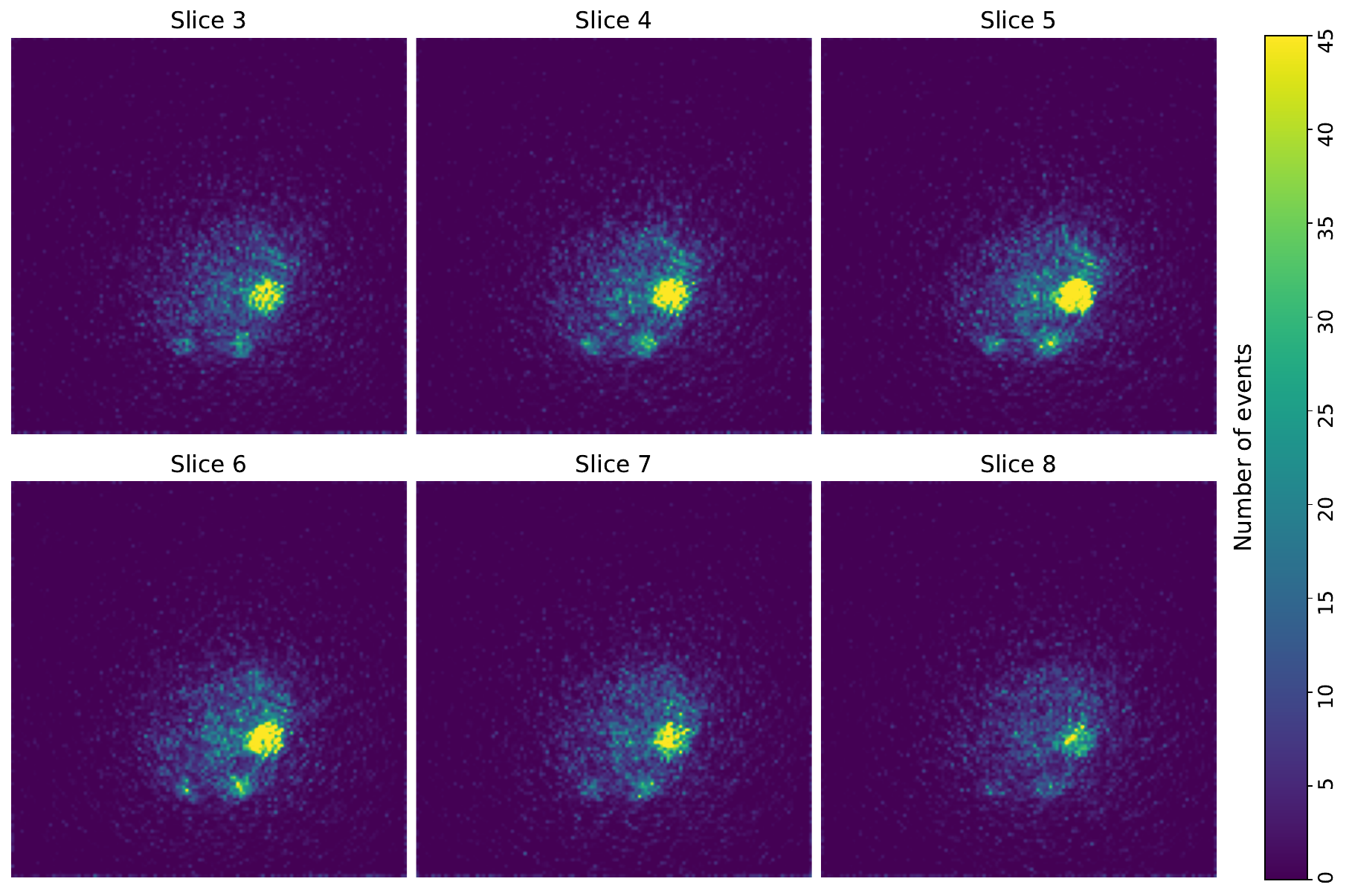}
\caption{Reconstructed detected-activity image for the experimental J-PET phantom data. Central slices show the three \(^{44}\)Sc-filled spheres identified from retained triple coincidences; low-level activity in nominally cold regions is consistent with false or misassociated triple events.}\label{fig:jactivity}
\end{figure}

The voxel-wise effective annihilation-rate map obtained with the posterior-weighted Gamma approximation is shown in Figure~\ref{fig:jrate}. With 234\,375 voxels, approximately \(6.2\times10^4\) observed detector-time channels, and about \(3.64\times10^5\) retained events, the effective rate field was estimated in approximately 3~s on a single CPU without explicit parallelization.

\begin{figure}[htbp]
\centering
\includegraphics[width=4in]{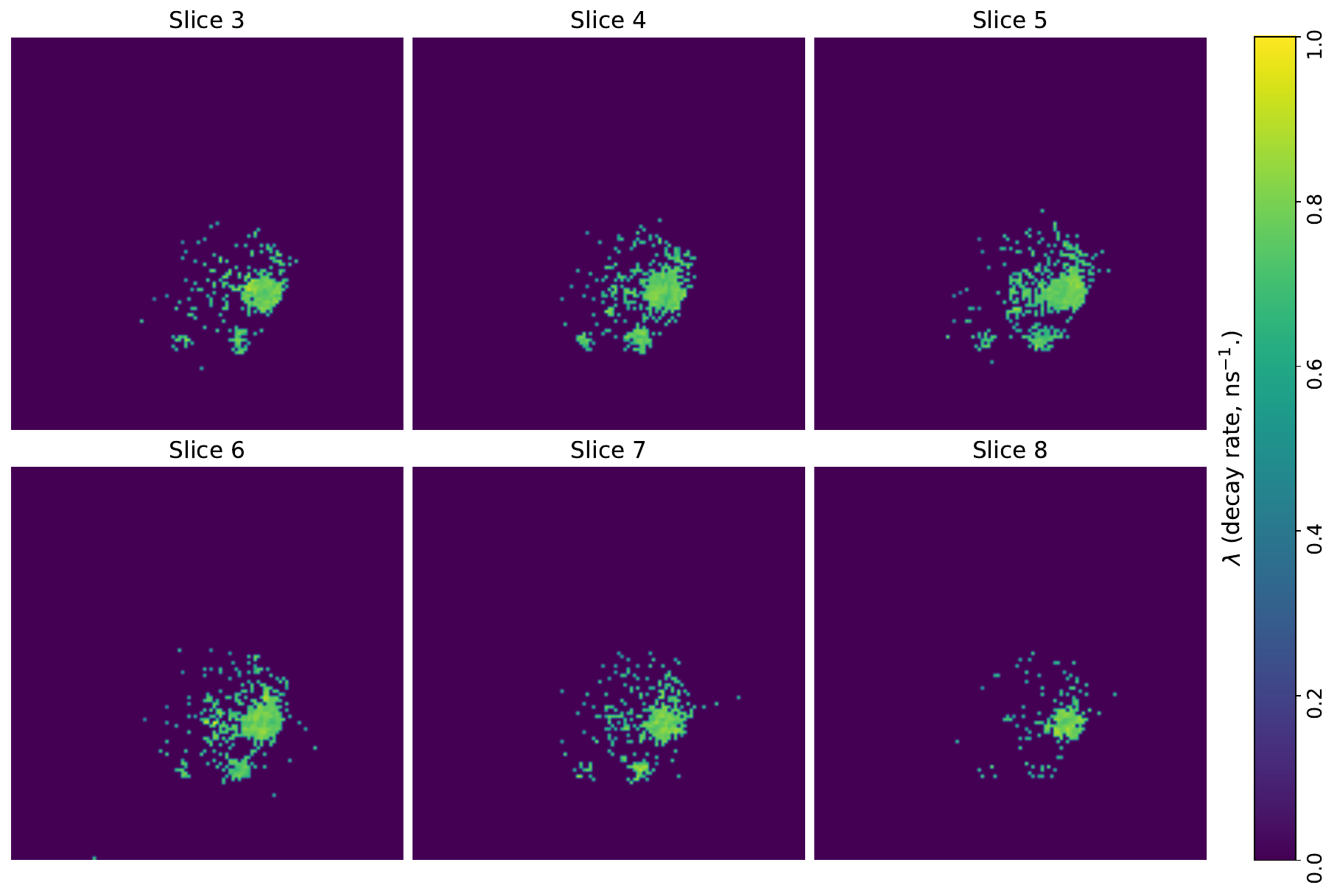}
\caption{Voxel-wise effective annihilation-rate map reconstructed from the experimental J-PET phantom dataset using the posterior-weighted Gamma approximation. Rates are reported in ns\(^{-1}\); the map is intended as a single-component surrogate summary of the observed timing signal.}\label{fig:jrate}
\end{figure}

Figure~\ref{fig:jhist} summarizes the whole-volume distribution of effective rates and reciprocal posterior-mean rates. The median effective rate was \(\tilde\lambda=0.77\)~ns\(^{-1}\), corresponding to \(\tilde{\tau}_{\text{eff}}=1.298\)~ns. This value is of the same order as the approximately 1.10~ns effective lifetime reported in the reference J-PET phantom analysis \cite{das_first_2025}. However, the comparison should be interpreted qualitatively rather than as a direct validation because the present analysis uses a single-component surrogate model and does not include explicit random-triple subtraction or detector-response convolution.

\begin{figure}[htbp]
\centering
\includegraphics[width=4in]{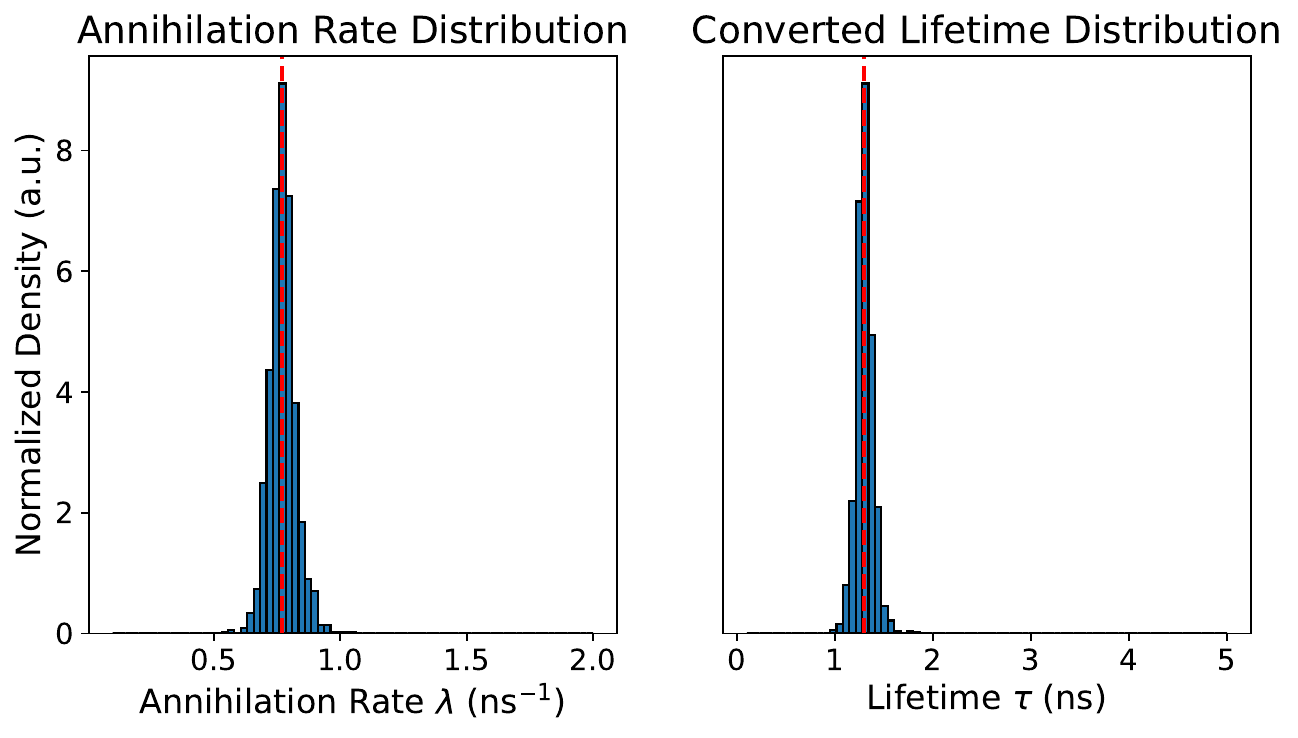}
\caption{Whole-volume histogram of voxel-wise effective annihilation rates for the experimental dataset. The corresponding reciprocal posterior-mean rates provide the effective lifetime summary discussed in the text.}\label{fig:jhist}
\end{figure}

\section{Discussion}\label{sec:discussion}

The results support the main computational idea of this work. Restricting the forward model to detector-time channels that are actually observed makes fully three-dimensional reconstruction feasible for sparse triple-coincidence data. In both simulation and experiment, the partial system matrix produced stable detected-activity reconstructions and enabled fast voxel-wise lifetime estimation. For realistic scanners the number of physically possible detector-time combinations is extremely large, while only a small fraction are populated during a finite acquisition. Conditioning the system matrix on observed channels therefore reduces memory and computational cost while preserving the Poisson data model.

The event-to-voxel weighting formulation provides a practical way to incorporate list-mode lifetime measurements into voxel-wise estimation. Because the source voxel of each event is unknown, events are distributed across voxels according to their posterior probabilities under the system matrix and the reconstructed detected activity. This fractional attribution produces weighted counts and weighted lifetime sums that can be accumulated directly from list-mode data. In practice this allows stable estimation even when individual voxels contain relatively few detected triple coincidences.

The conjugate Bayesian estimator then converts these weighted statistics into voxel-wise rate estimates through a closed-form Gamma update. Because the posterior is analytic, the lifetime map can be estimated without iterative optimization. In the simulation study this produced estimates comparable to the L-BFGS-B maximum-likelihood solution while requiring substantially less computation time. This computational advantage becomes more important as the reconstruction grid grows and the number of voxels increases.

The lifetime model used here is intentionally simple. Recent experimental work has shown that realistic positronium signals may require multi-component lifetime models, explicit treatment of detector timing response, and corrections for random triple coincidences \cite{steinberger_positronium_2024,mercolli_first_2025,mercolli_vivo_2026}. The present exponential model therefore provides a compact representation of the observed timing signal rather than a full physical description of positronium decay.

The Bayesian update should also be interpreted in this context. The event-to-voxel weights are computed from the channel information and the detected-activity estimate and are then treated as fixed when estimating the lifetime field. The resulting posterior therefore reflects uncertainty within this surrogate model rather than the full physical uncertainty associated with detector calibration, background contamination, or model mismatch.

The experimental results demonstrate that the framework remains practical for realistic datasets. Activity reconstruction identifies the regions containing \(^{44}\)Sc activity in the phantom, and the effective-rate map can be computed in a few seconds for a reconstruction grid containing more than two hundred thousand voxels. These results indicate that the partial system matrix and analytic lifetime update provide a scalable computational foundation for three-dimensional positronium lifetime imaging.

An important probabilistic distinction is that the partial system matrix \(H_{c,j}=P(C=c\mid J=j,\text{retained})\) is a forward channel model used for activity reconstruction, whereas \(\pi_{k,j}=P(J=j\mid C=c_k,\tilde f)\) is the reverse voxel-attribution probability obtained from Bayes' rule. The marginal retained-event mixture weights satisfy \(p_j\propto \tilde f_j\) because \(\tilde f_j\) already represents retained detected activity. Keeping these objects distinct, namely \(\tilde f_j\), \(H_{c,j}\), and \(\pi_{k,j}\), is essential for a coherent probabilistic interpretation of the two-stage estimator.

\section{Conclusions}\label{sec:conclusion}

We developed a three-dimensional statistical framework for positronium lifetime imaging that combines a TOF-aware partial system matrix, event-to-voxel fractional attribution, and a conjugate Bayesian lifetime estimator. Restricting the system matrix to observed detector-time channels preserves the Poisson forward model while making fully three-dimensional reconstruction computationally practical for sparse triple-coincidence data. The weighting formulation provides a principled way to incorporate list-mode lifetime measurements without explicit event assignment, and the conjugate Bayesian update yields closed-form voxel-wise estimates that avoid iterative optimization. Simulation results show that the analytic estimator achieves accuracy comparable to maximum likelihood at substantially lower computational cost. Application to a J-PET phantom dataset demonstrates that the framework scales to large volumetric reconstructions and recovers the expected activity structure. Together, these results establish a practical computational foundation for large-scale three-dimensional positronium lifetime imaging.

\section*{List of abbreviations}
EM: expectation-maximization; J-PET: Jagiellonian positron emission tomography; L-BFGS-B: limited-memory Broyden--Fletcher--Goldfarb--Shanno with bound constraints; NEMA: National Electrical Manufacturers Association; PET: positron emission tomography; PLI: positronium lifetime imaging; TOF: time of flight.

\backmatter

\bmhead{Supplementary Information}
Additional methodological details, theoretical derivations, and supplementary figures are provided in Additional file~1.

\section*{Declarations}

\subsection*{Ethics approval and consent to participate}
Not applicable.

\subsection*{Consent for publication}
Not applicable.

\subsection*{Availability of data and materials}
The simulated datasets generated during the current study and the processed phantom datasets analysed during the current study are available from the corresponding author on reasonable request. Raw J-PET phantom data are not publicly available because access may require collaboration approval, but they may be available from the corresponding author on reasonable request and with permission of the data owners.

\subsection*{Competing interests}
PM is an inventor on patents related to positronium imaging (Poland PL~227658, Europe EP~3039453, United States US~9,851,456). All other authors declare that they have no competing interests.

\subsection*{Funding}
C.-M.K. was supported in part by NIH grant R01-EB029948. H.-H.H. was supported in part by NSF grants DMS-1924792 and DMS-2318925. M.D., P.M., A.P., S.S., and E.S. were supported in part by the National Science Centre of Poland through grants 2021/42/A/ST2/00423, 2021/43/B/ST2/02150, 2022/47/I/NZ7/03112, and 2023/50/E/ST2/0057; by the Ministry of Science and Higher Education through grant SPUB/SP/627733/2025; by the SciMat and qLife Priority Research Areas budget under the Excellence Initiative -- Research University program at Jagiellonian University; and by the European Research Council through Advanced Grant POSITRONIUM no.~101199807. The funding bodies had no role in the design of the study, the analysis of the data, or the writing of the manuscript.

\subsection*{Authors' contributions}
BU and HHH developed the statistical framework. BU implemented the reconstruction software and performed the simulations. BU and GG prepared the computational experiments and figures. MD, AP, SS, PM, ES, and JC contributed to the experimental data acquisition, curation, and interpretation. CMK contributed to the medical-imaging interpretation and manuscript revision. HHH supervised the project. BU and HHH drafted the manuscript. All authors critically revised the manuscript and approved the final version.

\subsection*{Acknowledgements}
Not applicable.

\subsection*{Additional files}
\textbf{Additional file 1: Additional\_file\_1.pdf}\\
\textbf{File format:} .pdf\\
\textbf{Title:} Supplementary material for ``A Conjugate Bayesian Framework for Fast 3D Positronium Lifetime Estimation with a Partial System Matrix''\\
\textbf{Description:} Contains additional mathematical derivations and proofs, supplementary figures, and extended discussion supporting the proposed framework.

\bibliography{references}

\newpage

\begin{center}
\appendix{\Large \textbf{SUPPLEMENTARY}}
\end{center}

\renewcommand{\thefigure}{S\arabic{figure}}
\renewcommand{\thetable}{S\arabic{table}}
\renewcommand{\theequation}{S\arabic{equation}}
\renewcommand{\thesection}{S\arabic{section}}

\section{Statistical properties of the posterior-weighted Gamma estimator}
\subsection{Latent voxel mixture model under the partial system matrix}

Bayesian model-based clustering represents unknown cluster labels through latent allocation variables and bases inference on posterior allocation probabilities or Bayes-optimal partitions \cite{melnykov_distribution_2013,rastelli_optimal_2018,aragam_identifiability_2020,ascolani_clustering_2023}. We use the same construction here, with voxels playing the role of latent clusters and retained triple events playing the role of observations.

For the theoretical results below, we work in an idealized fixed-design setting in which the partial system matrix is defined on a fixed finite channel set \(\mathcal C_+\) and the voxel domain is restricted to the positive-support subset
\begin{equation}
\mathcal V_+ = \{j : \tilde f_j > 0\},
\qquad
J = |\mathcal V_+| < \infty.
\end{equation}
Throughout this supplement, indices \(j=1,\dots,J\) refer only to voxels in \(\mathcal V_+\). Zero-probability voxels and the random growth of the observed channel set with sample size are outside the scope of the consistency statements below.

Throughout this supplement, the partial system matrix is the forward conditional distribution
\begin{equation}
H_{c,j}=P(C=c\mid Z=j,\text{retained}),
\qquad
\sum_{c\in\mathcal C_+}H_{c,j}=1.
\label{eq:supp_forwardH}
\end{equation}
Let
\begin{equation}
p_j=\frac{\tilde f_j}{\sum_{\ell=1}^{J}\tilde f_\ell}
\label{eq:supp_pj}
\end{equation}
be the retained-event probability that a randomly selected retained triple originated from voxel \(j\). Conditional on \(\tilde f\), we model the retained event \((C_k,\tau_k)\) through an unobserved source voxel \(Z_k\in\{1,\dots,J\}\):
\begin{align}
Z_k \mid \tilde f &\sim \mathrm{Categorical}(p_1,\dots,p_J), \label{eq:supp_latentZ}\\
C_k \mid Z_k=j,\tilde f &\sim H_{\cdot,j}, \label{eq:supp_CgivenZ}\\
\tau_k \mid Z_k=j,\lambda &\sim \mathrm{Exp}(\lambda_j), \qquad \tau_k\ge 0. \label{eq:supp_TgivenZ}
\end{align}
The channel and timing variables are conditionally independent given \(Z_k\) and \(\lambda\). Therefore the retained-data density is
\begin{equation}
p(c,\tau\mid p,\lambda)=
\sum_{j=1}^{J} p_j H_{c,j}\lambda_j e^{-\lambda_j\tau}.
\label{eq:supp_jointdensity}
\end{equation}

Define the unnormalized source score
\begin{equation}
w_{k,j}=H_{c_k,j}\tilde f_j,
\qquad
\hat y_{c_k}=\sum_{\ell=1}^{J}w_{k,\ell},
\label{eq:supp_unnormscore}
\end{equation}
and the normalized channel-posterior weight
\begin{equation}
\pi_{k,j}=\frac{w_{k,j}}{\hat y_{c_k}}.
\label{eq:supp_piweights}
\end{equation}
\begin{remark}[Marginal versus channel-conditioned voxel probabilities]\label{rem:supp_p_vs_pi}
Equation~\eqref{eq:supp_pj} defines the marginal source-voxel probability among retained events. By contrast, Eq.~\eqref{eq:supp_piweights} is the channel-conditioned posterior allocation probability. Therefore the factor \(H_{c,j}\) belongs in \(\pi_{k,j}=P(Z_k=j\mid C_k=c_k,\tilde f)\), not in the marginal mixture weight \(p_j\). If one instead parameterizes the model by pre-detection activity \(f_j\) and an unnormalized retained-event kernel \(G_{c,j}=P(C=c,\text{retained}\mid Z=j)\), then the retained-event source probability is proportional to \(s_j f_j\) with \(s_j=\sum_c G_{c,j}\), while the channel-conditioned posterior becomes
\[
P(Z=j\mid C=c,\text{retained},f)
=
\frac{G_{c,j}f_j}{\sum_{\ell}G_{c,\ell}f_\ell}
=
\frac{H_{c,j}\tilde f_j}{\sum_{\ell}H_{c,\ell}\tilde f_\ell}.
\]
Thus even in the pre-detection parameterization, the channel factor enters only after conditioning on \(C=c\).
\end{remark}
\begin{theorem}[Bayesian interpretation of the partial-matrix weight]\label{thm:supp_bayesweight}
For each retained event \(k\), the score \(w_{k,j}=H_{c_k,j}\tilde f_j\) is the unnormalized posterior mass that the event originated from voxel \(j\) given its observed detector-time channel, and the normalized weight \(\pi_{k,j}\) is the exact posterior probability
\begin{equation}
\pi_{k,j}=P(Z_k=j\mid C_k=c_k,\tilde f).
\end{equation}
Moreover, if \(N_{c,j}\) denotes the unobserved number of retained events from voxel \(j\) recorded in channel \(c\), then
\begin{equation}
E[N_{c,j}\mid Y,\tilde f]=Y_c\,\pi_{c,j},
\qquad
\pi_{c,j}=\frac{H_{c,j}\tilde f_j}{\sum_{\ell}H_{c,\ell}\tilde f_\ell}.
\label{eq:supp_latentcount_expectation}
\end{equation}
\end{theorem}

\begin{proof}
By Bayes' rule and Eq.~\eqref{eq:supp_pj},
\[
P(Z_k=j\mid C_k=c_k,\tilde f)
=
\frac{P(C_k=c_k\mid Z_k=j,\tilde f)P(Z_k=j\mid \tilde f)}
{\sum_{\ell}P(C_k=c_k\mid Z_k=\ell,\tilde f)P(Z_k=\ell\mid \tilde f)}.
\]
Using Eq.~\eqref{eq:supp_CgivenZ},
\[
P(Z_k=j\mid C_k=c_k,\tilde f)
=
\frac{H_{c_k,j}\tilde f_j}{\sum_{\ell}H_{c_k,\ell}\tilde f_\ell}
=
\frac{w_{k,j}}{\hat y_{c_k}}.
\]
Hence \(w_{k,j}\) is the unnormalized posterior mass and \(\pi_{k,j}\) is the posterior probability. For the latent counts, Poisson splitting implies that conditional on \(Y_c\), the vector \((N_{c,1},\dots,N_{c,N_{\mathrm{vox}}})\) is multinomial with total \(Y_c\) and cell probabilities \(\pi_{c,j}\). Taking expectations gives Eq.~\eqref{eq:supp_latentcount_expectation}.
\end{proof}

Because \(\hat y_{c_k}\) does not depend on \(j\), maximizing the event likelihood with \(w_{k,j}\) or with \(\pi_{k,j}\) gives the same exact responsibility after the timing model is included. The exact posterior responsibility for event \(k\) is
\begin{equation}
r_{k,j}(p,\lambda)
=
P(Z_k=j\mid C_k=c_k,\tau_k,p,\lambda)
=
\frac{\pi_{k,j}\lambda_j e^{-\lambda_j\tau_k}}
{\sum_{\ell=1}^{J}\pi_{k,\ell}\lambda_{\ell}e^{-\lambda_{\ell}\tau_k}}.
\label{eq:supp_resp}
\end{equation}
Equation~\eqref{eq:supp_resp} is the mixture-model clustering responsibility specialized to the present detector-time setting.

\subsection{Gamma conjugacy and estimator construction}

\begin{proposition}[Conjugate conditional posterior]\label{prop:supp_gamma}
Under the hierarchical model in Eqs.~\eqref{eq:supp_latentZ}--\eqref{eq:supp_TgivenZ} with independent Gamma priors \(\lambda_j\sim\mathrm{Gamma}(\alpha_0,\beta_0)\), the exact conditional posterior given the latent labels \(Z_1,\dots,Z_{N_{\mathrm{ev}}}\) factorizes as
\begin{equation}
\lambda_j\mid Z,\tau
\sim
\mathrm{Gamma}\!\left(\alpha_0+N_j,\ \beta_0+S_j\right),
\end{equation}
where \(N_j=\sum_k I(Z_k=j)\) and \(S_j=\sum_k I(Z_k=j)\tau_k\).
\end{proposition}

\begin{proof}
The complete-data likelihood contributed by voxel \(j\) is
\[
\prod_{k:Z_k=j}\lambda_j e^{-\lambda_j\tau_k}
=
\lambda_j^{N_j}\exp(-\lambda_j S_j).
\]
Multiplying by the Gamma prior density gives
\[
p(\lambda_j\mid Z,\tau)\propto
\lambda_j^{\alpha_0+N_j-1}
\exp\!\left[-(\beta_0+S_j)\lambda_j\right],
\]
which is the stated Gamma distribution. Independence across \(j\) follows from the factorization of the prior and complete-data likelihood.
\end{proof}

\subsection{Theoretical guarantees}

The next theorem adapts Bayesian mixture-model clustering consistency to the finite-voxel retained-event model in Eq.~\eqref{eq:supp_jointdensity}. Related Bayes-optimal clustering and clustering-consistency results have been established in much broader mixture settings \cite{aragam_identifiability_2020,ascolani_clustering_2023}. Because the present model is finite-dimensional, we can give a direct proof tailored to the partial-system-matrix formulation.

For \(\varepsilon_p\in(0,1/J)\), define the compact simplex interior
\begin{equation}
\Delta_{\varepsilon_p}^{J-1}=
\left\{p\in[\varepsilon_p,1]^J:\ \sum_{j=1}^{J}p_j=1\right\},
\end{equation}
and let
\begin{equation}
\Theta=\Delta_{\varepsilon_p}^{J-1}\times[\underline\lambda,\overline\lambda]^J,
\qquad 0<\underline\lambda<\overline\lambda<\infty.
\end{equation}
For \(\theta=(p,\lambda)\in\Theta\), write
\begin{equation}
p_\theta(c,\tau)=\sum_{j=1}^{J}p_jH_{c,j}\lambda_j e^{-\lambda_j\tau},
\qquad c\in\mathcal C_+,\ \tau\ge 0.
\label{eq:supp_theta_density}
\end{equation}
The exact likelihood is
\begin{equation}
L_n(\theta)=\prod_{k=1}^{n}p_\theta(C_k,\tau_k).
\end{equation}

\begin{theorem}[Posterior consistency for mixed and unknown source voxels]\label{thm:supp_consistency}
Let \(J=|\mathcal V_+|<\infty\). Assume that the retained events \(\{(C_k,\tau_k)\}_{k=1}^{n}\) are independent and identically distributed from \(p_{\theta^0}\) for some interior point \(\theta^0=(p^0,\lambda^0)\in\Theta\). Suppose:
\begin{enumerate}
\item[(A1)] the map \(\theta\mapsto p_\theta\) is identifiable on \(\Theta\);
\item[(A2)] the prior on \(\Theta\) has a density that is continuous and strictly positive in a neighborhood of \(\theta^0\).
\end{enumerate}
Then for every \(\varepsilon>0\),
\begin{equation}
\Pi_n\!\left(\|\theta-\theta^0\|_\infty>\varepsilon\mid \mathcal D_n\right)\to 0
\end{equation}
in \(P_{\theta^0}\)-probability as \(n\to\infty\).
\end{theorem}

\begin{proof}
For \(\theta=(p,\lambda)\in\Theta\), define
\[
m_\theta(c,\tau)=\log p_\theta(c,\tau).
\]
Since \(\mathcal C_+\) is finite, for each \(c\in\mathcal C_+\) there exists \(h_c=\max_j H_{c,j}>0\). For any \(\theta\in\Theta\),
\[
\varepsilon_p\underline\lambda h_c e^{-\overline\lambda\tau}
\le p_\theta(c,\tau)
\le J\overline\lambda e^{-\underline\lambda\tau}.
\]
Hence there exist constants \(A,B<\infty\) such that
\[
|m_\theta(c,\tau)|\le A+B\tau
\qquad\text{for all }\theta\in\Theta,\ c\in\mathcal C_+,\ \tau\ge 0.
\]
Under \(p_{\theta^0}\), \(E_{\theta^0}[\tau]<\infty\), so \(A+B\tau\) is integrable. Because \(m_\theta(c,\tau)\) is continuous in \(\theta\) and \(\Theta\) is compact, the class \(\{m_\theta:\theta\in\Theta\}\) is Glivenko-Cantelli. Therefore
\begin{equation}
\sup_{\theta\in\Theta}
\left|
\frac1n\sum_{k=1}^{n}m_\theta(C_k,\tau_k)-M(\theta)
\right|\to 0
\label{eq:supp_uniform_lln_joint}
\end{equation}
almost surely, where \(M(\theta)=E_{\theta^0}[m_\theta(C,\tau)]\).

Now
\[
M(\theta)-M(\theta^0)
=
E_{\theta^0}\!\left[\log\frac{p_\theta(C,\tau)}{p_{\theta^0}(C,\tau)}\right]
=-\mathrm{KL}(p_{\theta^0}\,\|\,p_\theta)\le 0.
\]
By identifiability (A1), equality holds only at \(\theta=\theta^0\). Since \(M\) is continuous on compact \(\Theta\), for every \(\varepsilon>0\) there exists \(\delta_\varepsilon>0\) such that
\begin{equation}
\sup_{\|\theta-\theta^0\|_\infty>\varepsilon}M(\theta)
\le M(\theta^0)-2\delta_\varepsilon.
\label{eq:supp_joint_sep}
\end{equation}
Choose an open ball \(U_\varepsilon\subset\Theta\) centered at \(\theta^0\) with radius \(\varepsilon/2\) such that
\begin{equation}
\inf_{\theta\in U_\varepsilon}M(\theta)
\ge M(\theta^0)-\delta_\varepsilon/2.
\label{eq:supp_joint_neigh}
\end{equation}
By (A2), \(\Pi(U_\varepsilon)>0\).

On the event where Eq.~\eqref{eq:supp_uniform_lln_joint} is smaller than \(\delta_\varepsilon/2\), Eqs.~\eqref{eq:supp_joint_sep} and \eqref{eq:supp_joint_neigh} imply
\[
\sup_{\|\theta-\theta^0\|_\infty>\varepsilon}\log L_n(\theta)
\le n\{M(\theta^0)-3\delta_\varepsilon/2\},
\]
whereas for every \(\theta\in U_\varepsilon\),
\[
\log L_n(\theta)
\ge n\{M(\theta^0)-\delta_\varepsilon\}.
\]
Therefore
\[
\Pi_n(\|\theta-\theta^0\|_\infty>\varepsilon\mid\mathcal D_n)
\le
\frac{\exp(-n\delta_\varepsilon/2)\Pi(\Theta)}{\Pi(U_\varepsilon)}
\to 0.
\]
This proves posterior consistency.
\end{proof}

\begin{corollary}[Consistency of posterior voxel-allocation probabilities]\label{cor:supp_alloc_consistency}
Under the assumptions of Theorem~\ref{thm:supp_consistency}, for every \(\varepsilon>0\),
\begin{equation}
\Pi_n\!\left(
\max_{c\in\mathcal C_+}\max_{1\le j\le J}
\left|\pi_j(c;\theta)-\pi_j(c;\theta^0)\right|>\varepsilon
\,\mid\,\mathcal D_n
\right)\to 0,
\end{equation}
where
\begin{equation}
\pi_j(c;\theta)=\frac{p_jH_{c,j}}{\sum_{\ell=1}^{J}p_\ell H_{c,\ell}}.
\end{equation}
\end{corollary}

\begin{proof}
For each observed channel \(c\in\mathcal C_+\), the denominator \(\sum_\ell p_\ell H_{c,\ell}\) is strictly positive on \(\Delta_{\varepsilon_p}^{J-1}\). Hence \((c,\theta)\mapsto \pi_j(c;\theta)\) is continuous on the compact set \(\mathcal C_+\times\Theta\). Uniform continuity then implies that for every \(\varepsilon>0\) there exists \(\eta>0\) such that \(\|\theta-\theta^0\|_\infty<\eta\) implies
\[
\max_{c\in\mathcal C_+}\max_j |\pi_j(c;\theta)-\pi_j(c;\theta^0)|<\varepsilon.
\]
The claim follows from Theorem~\ref{thm:supp_consistency}.
\end{proof}

\begin{corollary}[Plug-in posterior consistency for the two-stage estimator]\label{cor:supp_plugin_consistency}
Assume the conditions of Theorem~\ref{thm:supp_consistency} hold, except that \(p^0\) is not given a prior and is replaced by an estimator \(\hat p_n\) satisfying \(\hat p_n\to p^0\) in \(P_{\theta^0}\)-probability. Let \(\Pi_n^{\mathrm{plug}}(\cdot\mid\mathcal D_n,\hat p_n)\) be the posterior on \(\lambda\in[\underline\lambda,\overline\lambda]^J\) based on
\begin{equation}
L_n^{\mathrm{plug}}(\lambda;\hat p_n)
=
\prod_{k=1}^{n}\left[\sum_{j=1}^{J}\hat p_{n,j}H_{C_k,j}\lambda_j e^{-\lambda_j\tau_k}\right].
\label{eq:supp_plugin_like}
\end{equation}
If the prior on \([\underline\lambda,\overline\lambda]^J\) is continuous and strictly positive near \(\lambda^0\), then for every \(\varepsilon>0\),
\begin{equation}
\Pi_n^{\mathrm{plug}}\!\left(\|\lambda-\lambda^0\|_\infty>\varepsilon\mid \mathcal D_n,\hat p_n\right)\to 0
\end{equation}
in \(P_{\theta^0}\)-probability.
\end{corollary}

\begin{proof}
Define
\[
m_{p,\lambda}(c,\tau)=\log\left(\sum_{j=1}^{J}p_jH_{c,j}\lambda_j e^{-\lambda_j\tau}\right).
\]
By the same envelope argument used in the proof of Theorem~\ref{thm:supp_consistency},
\[
\sup_{p\in U,\lambda\in[\underline\lambda,\overline\lambda]^J}
\left|\frac1n\sum_{k=1}^{n}m_{p,\lambda}(C_k,\tau_k)-E_{\theta^0}[m_{p,\lambda}(C,\tau)]\right|\to 0
\]
for every compact neighborhood \(U\) of \(p^0\) inside \(\Delta_{\varepsilon_p}^{J-1}\). Since \(\hat p_n\to p^0\) in probability, with probability tending to one \(\hat p_n\in U\), and therefore
\[
\sup_{\lambda}
\left|\frac1n\log L_n^{\mathrm{plug}}(\lambda;\hat p_n)-M_{\mathrm{plug}}(\lambda)\right|\to 0
\]
in probability, where
\[
M_{\mathrm{plug}}(\lambda)=E_{\theta^0}[m_{p^0,\lambda}(C,\tau)].
\]
The function \(M_{\mathrm{plug}}\) is uniquely maximized at \(\lambda^0\) by the same KL-divergence argument as in Theorem~\ref{thm:supp_consistency}. Repeating the neighborhood comparison and prior-mass argument from that proof yields the stated result.
\end{proof}

\begin{corollary}[Consistency of the weighted-Gamma surrogate under asymptotically correct clustering]\label{cor:supp_eb}
Suppose the assumptions of Corollary~\ref{cor:supp_plugin_consistency} hold and let \(q_{k,j}\in[0,1]\) satisfy \(\sum_j q_{k,j}=1\) and
\begin{equation}
\frac1n\sum_{k=1}^{n}|q_{k,j}-I(Z_k=j)|(1+\tau_k)\to 0
\label{eq:supp_q_assumption}
\end{equation}
in \(P_{\theta^0}\)-probability for each \(j\). Define
\(
N_{j,n}^q=\sum_k q_{k,j}
\)
and
\(
S_{j,n}^q=\sum_k q_{k,j}\tau_k
\),
and form the Gamma approximation
\begin{equation}
\lambda_j\mid\mathcal D_n
\approx
\mathrm{Gamma}(\alpha_0+N_{j,n}^q,\beta_0+S_{j,n}^q).
\label{eq:supp_gammaapprox}
\end{equation}
Then
\begin{equation}
\frac{\alpha_0+N_{j,n}^q}{\beta_0+S_{j,n}^q}\to \lambda_j^0
\end{equation}
in probability, and the approximate posterior variance is \(O_p(n^{-1})\). In particular, Eq.~\eqref{eq:supp_gammaapprox} is consistent whenever the channel-posterior weights \(\pi_{k,j}\) satisfy Eq.~\eqref{eq:supp_q_assumption}; a sufficient condition is asymptotic channel separation, that is, \(P(Z_k=j\mid C_k)\to I(Z_k=j)\) in average probability.
\end{corollary}

\begin{proof}
Let \(N_{j,n}=\sum_k I(Z_k=j)\) and \(S_{j,n}=\sum_k I(Z_k=j)\tau_k\). By Eq.~\eqref{eq:supp_q_assumption},
\[
\frac{|N_{j,n}^q-N_{j,n}|}{n}
\le
\frac1n\sum_k |q_{k,j}-I(Z_k=j)|\to 0,
\]
and
\[
\frac{|S_{j,n}^q-S_{j,n}|}{n}
\le
\frac1n\sum_k |q_{k,j}-I(Z_k=j)|\tau_k\to 0
\]
in probability. The law of large numbers under the true hierarchical model gives
\[
\frac{N_{j,n}}{n}\to p_j^0,
\qquad
\frac{S_{j,n}}{n}\to \frac{p_j^0}{\lambda_j^0},
\]
so
\[
\frac{N_{j,n}^q}{n}\to p_j^0,
\qquad
\frac{S_{j,n}^q}{n}\to \frac{p_j^0}{\lambda_j^0}.
\]
Because \(p_j^0>0\),
\[
\frac{\alpha_0+N_{j,n}^q}{\beta_0+S_{j,n}^q}
=
\frac{N_{j,n}^q/n+\alpha_0/n}{S_{j,n}^q/n+\beta_0/n}
\to \lambda_j^0.
\]
The Gamma variance equals
\[
\frac{\alpha_0+N_{j,n}^q}{(\beta_0+S_{j,n}^q)^2}=O_p(n^{-1}),
\]
so the approximate posterior concentrates at \(\lambda_j^0\).
\end{proof}
\section{Estimation of total decays from detected triple coincidences}

The main reconstruction framework estimates $\tilde f_j$, the expected number of detected
triple coincidences from voxel $j$. In some settings it is useful to approximate the total
number of decays (including undetected events) from $\tilde f$ using geometric arguments.

As a simplifying approximation, treat each voxel as an isotropic photon emitter within an
open-ended cylindrical detector. When the object is centered in the scanner and not adjacent
to the axial openings, the probability that a single emitted photon reaches the curved wall
(rather than exiting through an open end) can be approximated from the solid angle subtended
by the two circular openings. Let $p_{\text{wall},j}$ denote this wall-intersection probability
for a voxel at location $j$.

Assuming that the three photons in a triple coincidence have approximately independent
geometric acceptance, the probability of detecting all three is $p_{\text{wall},j}^3$. Let
$N_j$ denote the expected total number of decays in voxel $j$ during the acquisition. Then
\begin{equation}
\tilde f_j \approx p_{\text{wall},j}^3\,N_j,
\qquad
N_j \approx \frac{\tilde f_j}{p_{\text{wall},j}^3}.
\end{equation}
If a decay rate is desired, divide $N_j$ by the acquisition duration.

To compute $p_{\text{wall},j}$, treat each voxel as an isotropic source. If a surface subtends
solid angle $\Omega$ at the source, the fraction of emission directed toward it is $\Omega/4\pi$.
For an open cylinder, the fraction escaping through the openings is
\begin{equation}
p_{\text{exit},j} = \frac{\Omega_{\text{top},j} + \Omega_{\text{bottom},j}}{4\pi},
\qquad
p_{\text{wall},j}=1-p_{\text{exit},j}.
\end{equation}
The solid angles $\Omega_{\text{top},j}$ and $\Omega_{\text{bottom},j}$ follow from surface
integration over the circular openings.

The resulting geometric correction and the corresponding error analysis are illustrated in
Figures~\ref{fig:f_map_real_est} and~\ref{fig:f_real_compare}.

\begin{figure}[htbp]
\centering{\includegraphics[width=4in]{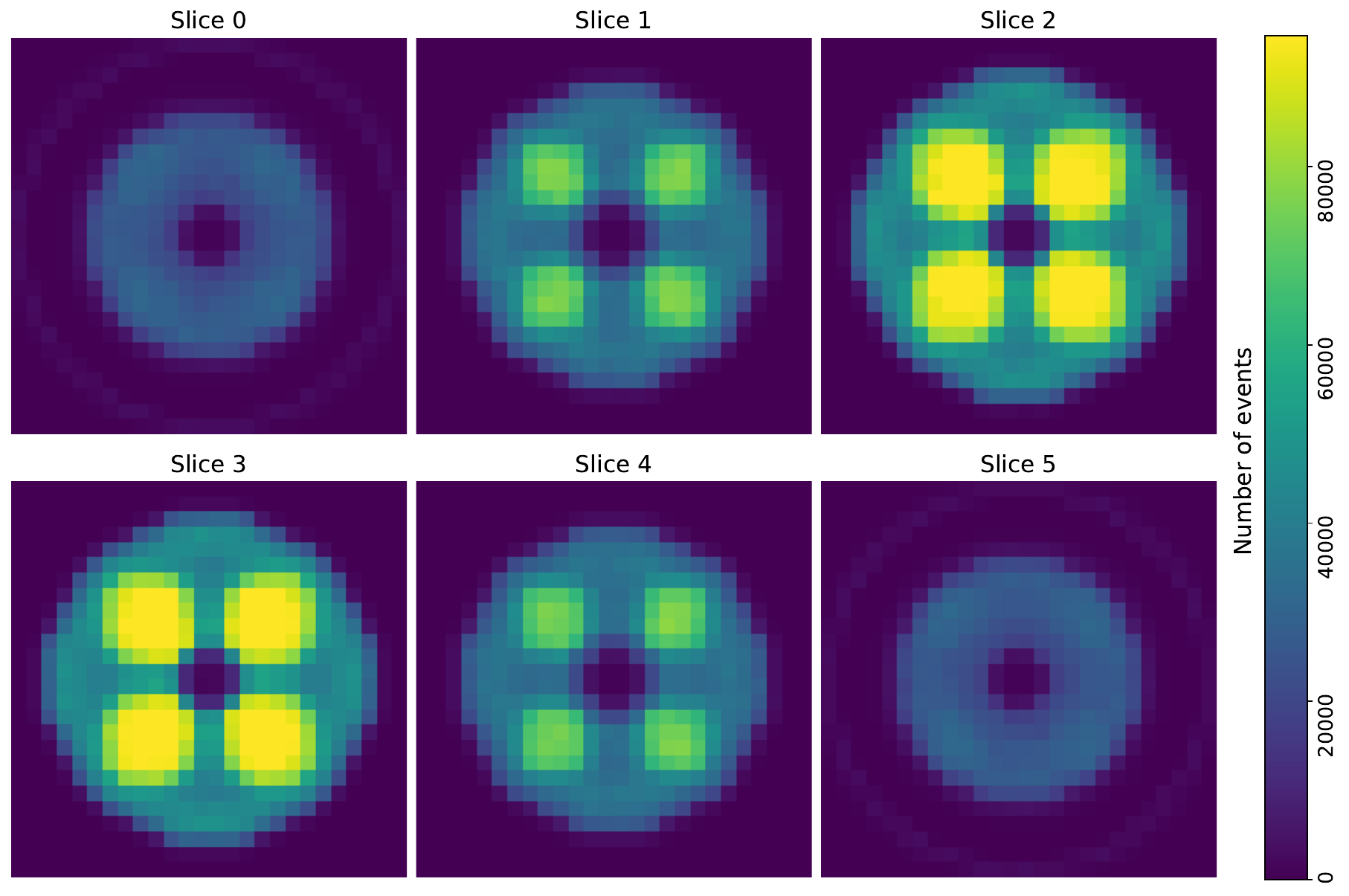}}
\caption{\label{fig:f_map_real_est} Activity estimate after applying the isotropic-emitter geometric correction. The slices show the inferred total-decay counts and visually match the ground-truth structure.}
\end{figure}

\begin{figure}[htbp]
\centering{\includegraphics[width=4in]{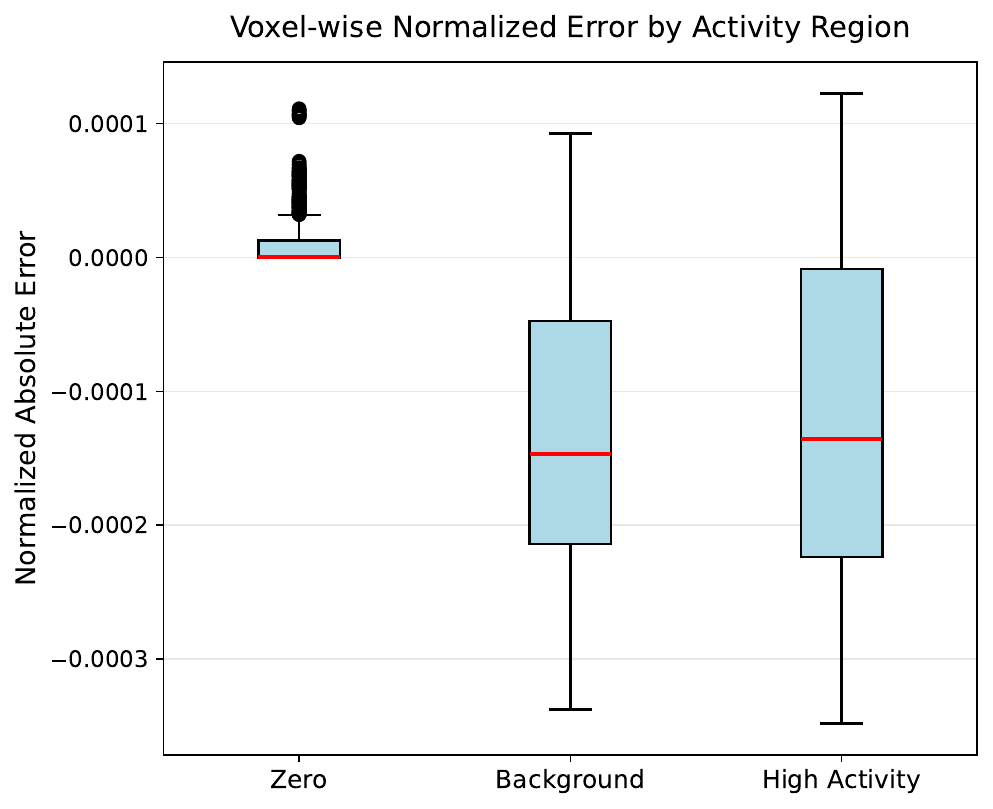}}
\caption{\label{fig:f_real_compare} Distribution of voxel-wise normalized absolute errors between the corrected activity estimate and ground truth, grouped by activity level.}
\end{figure}

\section{L-BFGS-B Estimation Results}

Figure~\ref{fig:lambda_freq} shows the decay-rate image reconstructed by maximum likelihood with L-BFGS-B (10 iterations). The map closely matches the ground-truth lifetime in Figure~\ref{fig:phantom_lambda}, with clear separation among regions of differing $\lambda$.

\begin{figure}[htbp]
\centering{\includegraphics[width=3.5in]{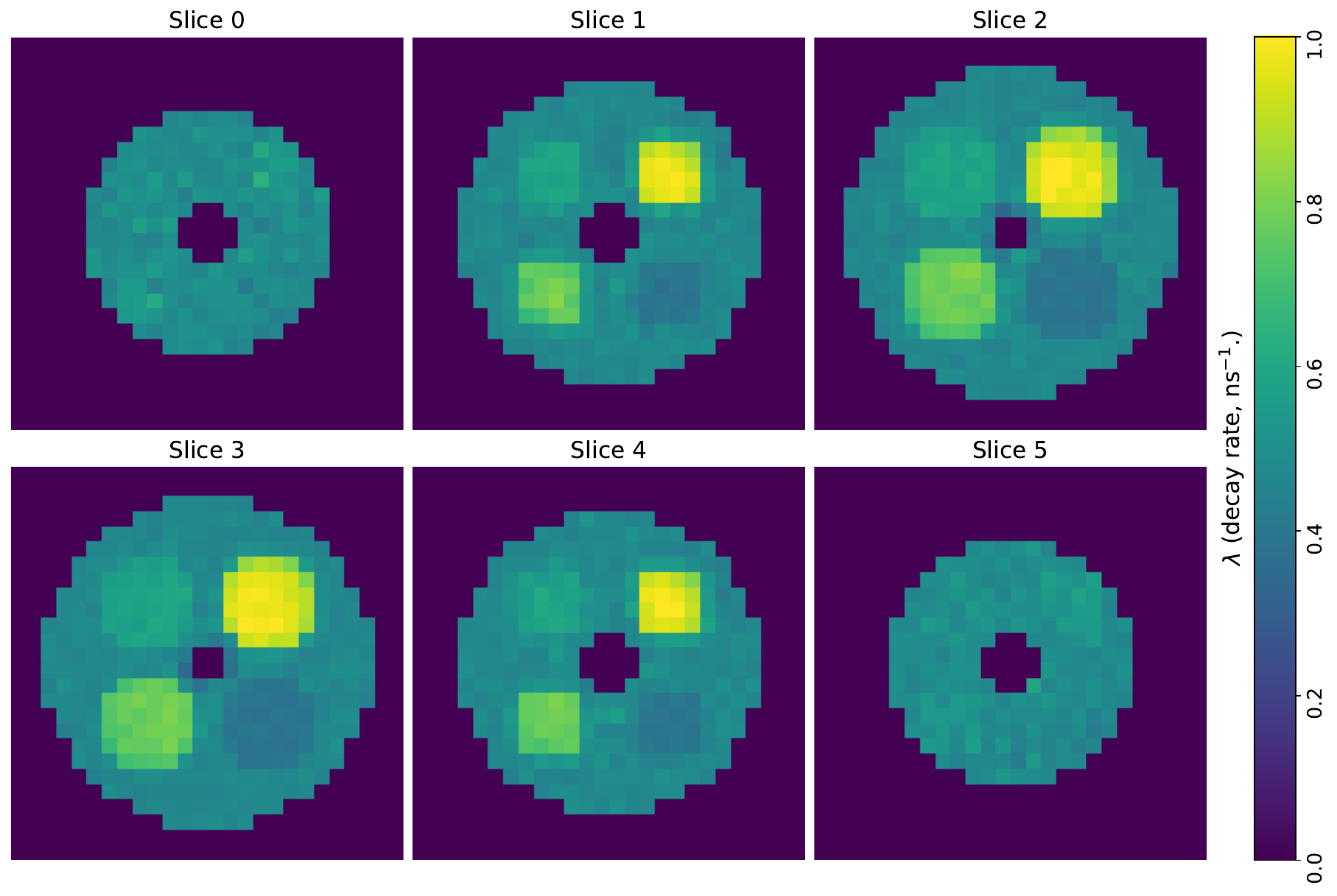}}
\caption{\label{fig:lambda_freq} Voxel-wise decay-rate map reconstructed
with L-BFGS-B (10 iterations). Five lifetime regions are delineated across the volume.}
\end{figure}

\begin{figure}[htbp]
\centering{\includegraphics[width=3.5 in]{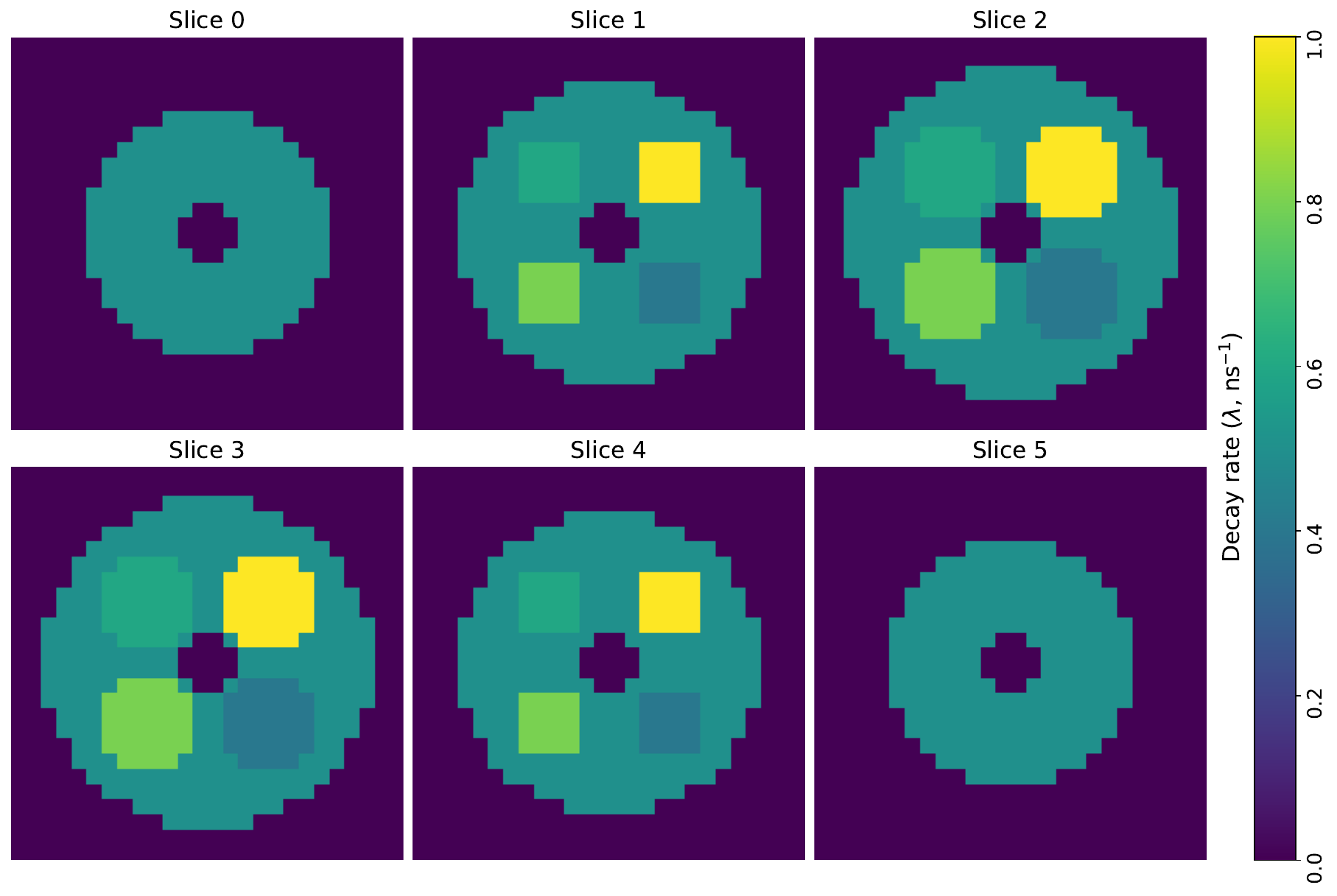}}
\caption{\label{fig:phantom_lambda} Ground-truth lifetime map. The color
scale shows positronium decay rates (ns$^{-1}$). Four inclusions with
increasing $\lambda$ are embedded in the uniform background to evaluate
sensitivity to lifetime variation.}
\end{figure}

\setcounter{figure}{0}
\setcounter{table}{0}
\setcounter{equation}{0}
\end{document}